\documentclass[a4paper,fleqn]{article}
\usepackage{a4wide}
\usepackage{amsmath,amssymb,amsthm}
\usepackage{thmtools}
\usepackage{tikz}
\usetikzlibrary{arrows,calc,math}
\usetikzlibrary{decorations.pathreplacing}
\usepackage{enumerate}
\usepackage{abstract}
\usepackage{authblk}

\usepackage{hyperref}
\usepackage[T1]{fontenc}
\usepackage{mathpazo}
\makeatletter\@ifpackageloaded{mathpazo}\@tempswatrue\@tempswafalse
\if@tempswa
  \DeclareFontFamily{OT1}{pzc}{}
  \DeclareFontShape{OT1}{pzc}{m}{it}{<-> s * [1.15] pzcmi7t}{}
  \DeclareMathAlphabet{\mathpzc}{OT1}{pzc}{m}{it}
\fi\makeatother
\usepackage{bm}
\usepackage[utf8]{inputenc}

\usepackage[sorting=nyt,giveninits=true,maxbibnames=7,backend=biber]{biblatex}
\makeatletter\@ifpackageloaded{biblatex}{%
  \usepackage{csquotes} 
  \bibliography{../../references}
  \renewbibmacro{in:}{%
    \ifentrytype{incollection}{\printtext{\bibstring{in}\intitlepunct}}{}}
  \renewbibmacro{publisher+location+date}{%
    \iflistundef{publisher}
      {\setunit*{\addcomma\space}}
      {\setunit*{\addcomma\space}}%
    \printlist{publisher}%
    \setunit*{\addcomma\space}%
    \printlist{location}%
    \setunit*{\addcomma\space}%
    \usebibmacro{date}%
    \newunit}
  \DeclareFieldFormat[article]{pages}{#1\isdot}
  \DeclareFieldFormat[article,incollection,inproceedings,unpublished]{title}{#1\isdot}
  \DeclareFieldFormat[thesis]{title}{\mkbibemph{#1\isdot}}
  \DeclareFieldFormat[unpublished]{date}{(#1)\isdot}
  \DeclareFieldFormat[unpublished]{note}{#1\nopunct} 
  \DeclareFieldFormat[article]{journaltitle}{\mkbibemph{#1\isdot}}
  
  \AtEveryBibitem{%
    \ifentrytype{book}{}{
      \clearname{editor}
    }
  }
  \newbibmacro*{bbx:parunit}{%
    \ifbibliography
      {\setunit{\bibpagerefpunct}\newblock
       \usebibmacro{pageref}%
       \clearlist{pageref}%
       \setunit{\adddot\par\nobreak}}
      {}
  }
  \renewbibmacro*{doi+eprint+url}{%
    \usebibmacro{bbx:parunit}
    \iftoggle{bbx:doi}
      {\printfield{doi}}
      {}%
    \iftoggle{bbx:eprint}
      {\usebibmacro{eprint}}
      {}%
    \iftoggle{bbx:url}
      {\usebibmacro{url+urldate}}
      {}
  }
  \renewbibmacro*{eprint}{%
    \usebibmacro{bbx:parunit}
    \iffieldundef{eprinttype}
      {\printfield{eprint}}
      {\printfield[eprint:\strfield{eprinttype}]{eprint}}
  }
  \renewbibmacro*{url+urldate}{%
    \usebibmacro{bbx:parunit}
    \printfield{url}%
    \iffieldundef{urlyear}
      {}
      {\setunit*{\addspace}%
       \printtext[urldate]{\printurldate}}
  }
}{}\makeatother

\declaretheorem[numberwithin=section,refname={theorem,theorems},Refname={Theorem,Theorems}]{theorem}
\declaretheorem[sibling=theorem,style=definition]{definition}
\declaretheorem[sibling=theorem,name=Lemma]{lemma}
\declaretheorem[sibling=theorem,name=Proposition]{proposition}

\declaretheorem[sibling=theorem,name=Remark,style=definition]{remark}
\declaretheorem[sibling=theorem,name=Example,style=definition]{example}
\declaretheorem[sibling=theorem,name=Claim]{claim}
\declaretheorem[sibling=theorem,name=Conjecture]{conjecture}

\makeatletter\@ifpackageloaded{hyperref}{%
  \usepackage{xcolor}
  \definecolor{dark-red}{rgb}{0.4,0.15,0.15}
  \definecolor{dark-blue}{rgb}{0.15,0.15,0.4}
  \definecolor{medium-blue}{rgb}{0,0,0.5}
  \hypersetup{
    colorlinks,
    linkcolor={dark-red},
    citecolor={dark-blue},
    urlcolor={medium-blue}%
  }

}{}\makeatother

\usepackage{sectsty}
\sectionfont{\large\bfseries}
\subsectionfont{\normalsize\bfseries}

\providecommand{\abs}[1]{\lvert#1\rvert}
\providecommand{\Abs}[1]{\left\lvert#1\right\rvert}
\providecommand{\norm}[1]{\lVert#1\rVert}
\providecommand{\floor}[1]{\lfloor#1\rfloor}
\providecommand{\Floor}[1]{\left\lfloor#1\right\rfloor}
\newcommand{\Lang}[2][]{\mathcal{L}_{#1}(#2)}
\newcommand{\infw}[1]{%
  \ifcat\noexpand#1\relax\bm{#1}
  \else\mathbf{#1}\fi}          
\newcommand{\mirror}[1]{\widetilde{#1}}
\newcommand{\Parikh}[2][]{\mathcal{P}_{#1}(#2)}
\newcommand{\abexp}[2][]{\mathpzc{A\mkern-3mu e}_{#1}(#2)}
\newcommand{\qk}[1]{\mathcal{Q}_{#1}}
\newcommand{\qkl}[1]{\mathcal{Q}^+_{#1}}
\newcommand{\N}{\mathbb{N}}
\newcommand{\T}{\mathbb{T}}
\newcommand{\Z}{\mathbb{Z}}

\newcommand{\keywords}[1]{\par\noindent{\footnotesize{\em Keywords\/}: #1}}

\begin{document}
  \title{Abelian periods of factors of Sturmian words}
  \author[,1,2,3]{Jarkko Peltomäki\footnote{Corresponding author.\\E-mail address: \href{mailto:r@turambar.org}{r@turambar.org} (J. Peltomäki).}}
  \affil[1]{The Turku Collegium for Science and Medicine TCSM, University of Turku, Turku, Finland}
  \affil[2]{Turku Centre for Computer Science TUCS, Turku, Finland}
  \affil[3]{University of Turku, Department of Mathematics and Statistics, Turku, Finland}
  \date{}
  \maketitle
  \noindent
  \hrulefill
  \begin{abstract}
    \vspace{-1em}
    \noindent
    We study the abelian period sets of Sturmian words, which are codings of irrational rotations on a one-dimensional
    torus. The main result states that the minimum abelian period of a factor of a Sturmian word of angle $\alpha$ with
    continued fraction expansion $[0; a_1, a_2, \ldots]$ is either $tq_k$ with $1 \leq t \leq a_{k+1}$ (a multiple of a
    denominator $q_k$ of a convergent of $\alpha$) or $q_{k,\ell}$ (a denominator $q_{k,\ell}$ of a semiconvergent of
    $\alpha$). This result generalizes a result of Fici et al. stating that the abelian period set of the Fibonacci
    word is the set of Fibonacci numbers. A characterization of the Fibonacci word in terms of its abelian period set
    is obtained as a corollary.
    \vspace{1em}
    \keywords{Sturmian word, continued fraction, abelian equivalence, abelian period, singular word}
    \vspace{-1em}
  \end{abstract}
  \hrulefill

  \section{Introduction}
  let $w = w_0 w_1 \dotsm w_{\abs{w}-1}$ to be a finite word of length $\abs{w}$ composed of letters $w_0$, $w_1$,
  $\ldots$, $w_{\abs{w}-1}$. The word $w$ has period $p$ if $w_i = w_{i+p}$ for all $i$ with
  $0 \leq i \leq \abs{w} - 1 - p$. For example, the word $abaab$ has period $3$. Periods of words have been extensively
  studied; see, e.g., \cite[Ch.~8]{2002:algebraic_combinatorics_on_words}. One famous result is the Theorem of Fine and
  Wilf which states that if a word $w$ has two periods $p$ and $q$ and $\abs{w} \geq p + q - \gcd(p, q)$, then $w$ has
  period $\gcd(p, q)$ \cite{1965:uniqueness_theorems_for_periodic_functions}.  

  The majority of the research on periods has been about understanding the structure of periods of a single finite
  word. Much less attention has been paid to period sets. The period set of a (finite or infinite) given word is the
  set of minimum periods of all of its factors (subwords). For instance, the above word $abaab$ has proper factors $a$,
  $b$, $ab$, $ba$, $aa$, $aba$, $baa$, $aab$, $abaa$, and $baab$, so its period set is $\{1, 2, 3\}$. It seems that the
  only papers written on period sets are the 2009 seminal paper
  \cite{2009:least_periods_of_factors_of_infinite_words} of J. Currie and K. Saari and the 2012 preprint
  \cite{2012:least_periods_of_k-automatic_sequences} of D. Go{\v c} and J. Shallit. Currie and Saari study the period
  sets of infinite words. They show that the period set of the Thue-Morse word
  \cite{1999:the_ubiquitous_prouhet-thue-morse_sequence} is the set of positive integers, and they prove the following
  theorem on the period sets of Sturmian words, codings of irrational rotations on a one-dimensional torus.

  \begin{theorem}\label{thm:usual_period}\cite[Cor.~3]{2009:least_periods_of_factors_of_infinite_words}
    The period set of a Sturmian word of slope $\alpha$ having continued fraction expansion $[0; a_1, a_2, \ldots]$ is
    $\{\ell q_k + q_{k-1} : \text{$k \geq 0$, $\ell = 1, \ldots, a_{k+1}$}\}$, where the sequence $(q_k)$ is the
    sequence of denominators of convergents of $\alpha$.\footnote{Sturmian words are binary, and the slope of a Sturmian word is the (irrational) frequency of the letter having frequency less than $\tfrac12$.}
  \end{theorem}

  When \autoref{thm:usual_period} is applied to the Fibonacci word whose slope has continued fraction expansion
  $[0; 2, \overline{1}]$ (by a bar, we indicate a repeating pattern), we obtain the following nice theorem.

  \begin{theorem}\label{thm:usual_fibonacci}\cite[Cor.~4]{2009:least_periods_of_factors_of_infinite_words}
    The period set of the Fibonacci word is the set of Fibonacci numbers.
  \end{theorem}

  In the 2016 paper \cite{2016:abelian_powers_and_repetitions_in_sturmian_words} by the author and others, the period
  set of the Fibonacci word was studied with a generalized notion of a period called an abelian period, and an analogue
  of \autoref{thm:usual_fibonacci} was obtained in this generalized setting. The goal of this paper is to extend this
  result to all Sturmian words and obtain an analogue of \autoref{thm:usual_period} for abelian periods.

  Sturmian words are central objects in combinatorics on words. Their study was initiated in the 1940 paper
  \cite{1940:symbolic_dynamics_ii} by M. Morse and G. Hedlund. Sturmian words often exhibit extremal behavior among
  infinite words, and their properties have links to other areas of mathematics like number theory and discrete
  geometry. They admit many interesting combinatorial and dynamical generalizations. See
  \cite{1996:recent_results_on_sturmian_words,2009:episturmian_words_a_survey} and the references therein.

  Two words $u$ and $v$ are called abelian equivalent if one is obtained from the other by permuting letters. If $u_0$,
  $u_1$, $\ldots$, $u_{n-1}$ are abelian equivalent words of length $m$, then their concatenation
  $u_0 u_1 \dotsm u_{n-1}$ is called an abelian power of period $m$ and exponent $n$. For example,
  $aba \cdot baa \cdot aab$ is an abelian cube. This notion is a generalization of the concept of a power: a power is
  simply a repetition of the same word such as $aba \cdot aba \cdot aba$ (a cube). Recently it has been popular to
  generalize concepts and questions regarding ordinary powers and periods to this abelian setting. The foundational
  paper here is \cite{2011:abelian_complexity_of_minimal_subshifts}; see
  \cite{2016:abelian_powers_and_repetitions_in_sturmian_words} for additional references. For example, an appropriate
  generalization of the Theorem of Fine and Wilf was given in
  \cite{2006:fine_and_wilfs_theorem_for_abelian_periods,2013:abelian_periods_partial_words_and_an_extension,2016:an_abelian_periodicity_lemma}.
  These papers naturally contain the definition of an abelian period, which we shall give next; cf.
  \cite{2012:on_abelian_repetition_threshold}.
  
  Let $w$ be a finite word. Then $w$ has abelian period $m$ if $w$ is a factor of an abelian power $u_0 u_1 \dotsm
  u_{n-1}$ with $\abs{u_0} = \ldots = \abs{u_{n-1}} = m$. For example, the word $abaababa$, having minimum period $5$,
  has abelian periods $2$ and $3$ because it is a factor of the abelian powers $babaababab$ and $abaababaa$
  respectively. This indeed generalizes the concept of a period: a word has period $p$ if and only if it is a factor of
  some power of a word of length $p$.
  
  The abelian period set of an infinite word $\infw{w}$ is defined as the set of minimum abelian periods of its
  nonempty factors. As was done in \cite{2009:least_periods_of_factors_of_infinite_words} by Currie and Saari for the
  usual period set, we may now ask for a characterization of the abelian period set for a given word or class of words.
  For the Thue-Morse word, this is easy. The Thue-Morse word $\infw{t}$ is the fixed point of the substitution
  $0 \mapsto 01$, $1 \mapsto 10$ beginning with the letter $0$, and it is clear that $\infw{t}$ is an infinite
  concatenation of the words $01$ and $10$. Thus every factor of $\infw{t}$ has abelian period $2$. The minimum abelian
  period can equal $1$, but this happens only for finitely many factors because $000$ and $111$ do not occur in
  $\infw{t}$. Hence the abelian period set of $\infw{t}$ is $\{1, 2\}$. This should be compared with
  \cite[Thm.~2]{2009:least_periods_of_factors_of_infinite_words}: the period set of $\infw{t}$ is the set of positive
  integers.

  Characterizing the abelian period sets of Sturmian words is significantly harder. The following result was proved in
  \cite{2016:abelian_powers_and_repetitions_in_sturmian_words} for the Fibonacci word (which can be said to be the
  simplest Sturmian word). It should be compared with \autoref{thm:usual_fibonacci}.

  \begin{theorem}\label{thm:fibonacci}\cite[Thm.~6.9]{2016:abelian_powers_and_repetitions_in_sturmian_words}, \cite[Thm.~6.12]{2016:abelian_powers_and_repetitions_in_sturmian_words}
    The abelian period set of the Fibonacci word is the set of Fibonacci numbers.
  \end{theorem}

  What are then the abelian period sets of other Sturmian words? By simply replacing the word ``period set'' with
  ``abelian period set'' in the statement of \autoref{thm:usual_period} yields a false statement. Indeed, it was
  observed in \cite[Remark~6.11]{2016:abelian_powers_and_repetitions_in_sturmian_words} that, for example, the factor
  \begin{equation*}
    00101 \cdot 001001001010010010010100100100 \cdot 10100,
  \end{equation*}
  of a Sturmian word of slope $[0;\overline{2,1}]$ has minimum abelian period $6$, which is not of the form
  $\ell q_k + q_{k-1}$ for this slope. This example showed that the proof of \autoref{thm:fibonacci} in
  \cite{2016:abelian_powers_and_repetitions_in_sturmian_words} is not generalizable to all Sturmian words. In this
  paper, we present new ideas that work for all Sturmian words and prove the following result, which is the main result
  of this paper.

  \begin{theorem}\label{thm:main_intro}
    If $m$ is the minimum abelian period of a nonempty factor of a Sturmian word of slope $\alpha$ having continued
    fraction expansion $[0; a_1, a_2, \ldots]$, then either $m = tq_k$ for some $k \geq 0$ and some $t$ such that
    $1 \leq t \leq a_{k+1}$ or $m = \ell q_k + q_{k-1}$ for some $k \geq 1$ and some $\ell$ such that
    $1 \leq \ell < a_{k+1}$, where the sequence $(q_k)$ is the sequence of denominators of convergents of $\alpha$.
  \end{theorem}

  \autoref{thm:main_intro} essentially says that certain multiples of the numbers $q_k$ must also be allowed as minimum
  abelian periods. \autoref{thm:main_intro} implies \autoref{thm:fibonacci} (see the end of \autoref{sec:main_proofs}).

  Notice that \autoref{thm:main_intro} does not characterize the abelian period sets completely. Indeed, we shall see
  at the end of \autoref{sec:main_proofs} that the set of possible minimum abelian periods given by
  \autoref{thm:main_intro} can be unnecessarily large. The complete answer seems to depend on the slope $\alpha$ in a
  complicated way. To us \autoref{thm:main_intro} seems to be the best result obtainable without additional assumptions
  about the arithmetical nature of the slope $\alpha$. Because of this, we leave the complete characterization open.

  \autoref{thm:main_intro} allows an interesting characterization of the Fibonacci subshift, the shift orbit closure of
  the Fibonacci word, as the Sturmian subshift of slope $\alpha$ whose language $\Lang{\alpha}$ has the following
  property: the minimum abelian period of each $w \in \Lang{\alpha}$ is a denominator of a convergent of $\alpha$. See
  \autoref{thm:fibonacci_characterization} at the end of \autoref{sec:main_proofs}. This adds yet another property to
  the rather long list of extremal properties of the Fibonacci word
  \cite{2008:on_extremal_properties_of_the_fibonacci_word,2013:some_extremal_properties_of_the_fibonacci_word,2014:lyndon_words_and_fibonacci_numbers}.

  Even though the problems considered in this paper have their background in combinatorics and formal languages, a
  large part of the proofs are completely number-theoretic. It was already observed in
  \cite{2016:abelian_powers_and_repetitions_in_sturmian_words} (and independently in
  \cite{2013:some_properties_of_abelian_return_words}) that abelian powers and their exponents in Sturmian words can be
  studied effectively using continued fractions; in fact it is almost impossible to do without them. We continue to use
  this powerful tool. We give combinatorial arguments to derive a certain inequality which must hold if a given number
  is the minimum abelian period of some factor. Then we proceed to study the inequality using continued fractions with
  little combinatorics involved. Some of the intermediate results presented could be of independent interest in the
  theory of continued fractions.

  The paper is organized as follows. In \autoref{sec:preliminaries}, we give the necessary definitions and background
  information on continued fractions, Sturmian words, and abelian equivalence. Auxiliary results needed for the main
  proofs are then presented in \autoref{sec:abelian_exponents}. The central proof ideas and derivation of the main
  inequality are given in \autoref{sec:main_inequality}; the actual proofs of the main results are presented in
  \autoref{sec:main_proofs}. We conclude the paper by briefly considering the so-called minimum $k$-abelian periods of
  factors of Sturmian words in \autoref{sec:k-abelian}; this is a further generalization of the notion of a period.
  
  \section{Preliminaries}\label{sec:preliminaries}
  We shall use standard notions and notation from combinatorics on words. These are found in, e.g.,
  \cite{2002:algebraic_combinatorics_on_words}, and we briefly repeat what we need here.

  An \emph{alphabet} is a finite nonempty set of \emph{letters}. A \emph{word} $a_0 a_1 \dotsm a_{n-1}$ of length $n$
  over $A$ is a finite sequence of letters of $A$. We refer to the empty word with the symbol $\varepsilon$. The length
  of a word $w$ is denoted by $\abs{w}$. In this paper, we only consider binary words, and we take them to be over the
  alphabet $\{0, 1\}$. By $\abs{w}_0$ (resp. $\abs{w}_1$), we refer to the number of letters $0$ (resp. $1$) in the
  word $w$. An \emph{infinite word} $\infw{w}$ is a map from $\N$ to an alphabet $A$, and we write, as is usual,
  $\infw{w} = a_0 a_1 \dotsm$ with $a_i \in A$ (we always index from $0$). We refer to infinite words in boldface
  symbols. Many of the notions given here extend naturally to infinite words.

  Given two words $u$ and $v$, their product $uv$ is formed by concatenating their letters. A word $z$ is a
  \emph{factor} of the word $w$ if $w = uzv$ for some words $u$ and $v$. If $u = \varepsilon$ (resp.
  $v = \varepsilon$), then $z$ is a \emph{prefix} (resp. \emph{suffix}) of $w$. The word $z$ is a \emph{proper} prefix
  (resp. proper suffix) if $z \neq \varepsilon$ and $v \neq \varepsilon$ (resp. $u \neq \varepsilon$). With $u^{-1}w$
  and $wv^{-1}$ we respectively refer to the words $zv$ and $uz$. By $w^n$, we mean the word $w \dotsm w$ where $w$ is
  repeated $n$ times. Such a word is called an \emph{$n$th power}, or a \emph{repetition}. If $w = uzv$, then we say
  that $z$ \emph{occurs} in $w$ in position $\abs{u}$. In other words, the position $\abs{u}$ defines an
  \emph{occurrence} of $z$ in $w$. When we say that a factor $z$ occurs in $w$ in \emph{phase} $n$ modulo $q$, we mean
  that $z$ occurs in $w$ in a position $i$ such that $i \equiv n \pmod{q}$.

  Let $w = a_0 a_1 \dotsm a_{n-1}$ with $a_i \in A$. As mentioned in the introduction, the word $w$ has \emph{period}
  $p$ if $a_i = a_{i+p}$ for all $i$ with $0 \leq i \leq n - 1 - p$. The \emph{reversal} $\mirror{w}$ of $w$ is defined
  to be the word $a_{n-1} \dotsm a_1 a_0$. The word $w$ is a \emph{palindrome} if $\mirror{w} = w$. If a word $u$ has
  $w$ as a prefix and a suffix and contains exactly two occurrences of $w$, then we say that $u$ is a \emph{complete
  first return to $w$}. A word $u$ is a \emph{complete first return to $w$ in the same phase} if $u$ has $w$ as a
  prefix and as a suffix, $\abs{u} \equiv 0 \pmod{\abs{w}}$, $u$ contains at least two occurrences of $w$, and if $w$
  occurs in $u$ in position $i$ such that $i \equiv 0 \pmod{\abs{w}}$, then $i = 0$ or $i = \abs{u} - \abs{w}$. For
  example, the word $01001$ is a complete first return to $01$, but not a complete first return to $01$ in the same
  phase. The word $01001101$ is not a complete first return to $01$, but it is a complete first return to $01$ in the
  same phase.

  An infinite word $\infw{x}$ is \emph{recurrent} if each of its factors occur in it infinitely many times. Let
  $\infw{x} = x_0 x_1 \dotsm$ and $\infw{y} = y_0 y_1 \dotsm$ be two infinite words over an alphabet $A$. We endow
  $A^\N$, the set of infinite words over $A$, with the topology determined by the metric $d$ defined by
  \begin{equation*}
    d(\infw{x}, \infw{y}) = 2^{-k},
  \end{equation*}
  where $k$ is the least integer such that $x_k \neq y_k$ if $x \neq y$, and $k = \infty$ otherwise. The \emph{shift
  map} $T\colon A^\N \to A^\N$ is defined by setting the $n$th letter of $T\infw{x}$ to be the $(n+1)$th letter of
  $\infw{x}$. In other words, $T$ maps $(x_n)$ to $(x_{n+1})$. A \emph{subshift} is a closed and $T$-invariant subset
  of $A^\N$.
  
  Before defining abelian equivalence and the related concepts precisely, let us first recall some facts on continued
  fractions and define Sturmian words. For a more extensive introduction to continued fractions and Sturmian words, we
  refer the reader to \cite[Ch.~4]{diss:jarkko_peltomaki}. Good books on continued fractions are, e.g.,
  \cite{1979:an_introduction_to_the_theory_of_numbers,1997:continued_fractions} whereas
  \cite{2002:algebraic_combinatorics_on_words,2002:substitutions_in_dynamics_arithmetics_and_combinatorics} are good
  sources on Sturmian words.

  \subsection{Continued Fractions}
  Every irrational real number $\alpha$ has a unique infinite continued fraction expansion:
  \begin{equation}\label{eq:cf}
    \alpha = [a_0; a_1, a_2, a_3, \ldots] = a_0 + \dfrac{1}{a_1 + \dfrac{1}{a_2 + \dfrac{1}{a_3 + \ldots}}}
  \end{equation}
  with $a_0 \in \Z$ and $a_t \in \Z_+$ for $t \geq 1$. The numbers $a_i$ are called the \emph{partial
  quotients} of $\alpha$. The rational numbers $[a_0; a_1, a_2, a_3, \ldots, a_k]$, denoted by $p_k / q_k$, are called
  \emph{convergents} of $\alpha$. The convergents satisfy the following recurrences:
  \begin{alignat*}{4}
    p_0 = a_0, &\qquad p_1 = a_1 a_0 + 1, &\qquad p_k = a_k p_{k-1} + p_{k-2}, \qquad& k \geq 2, \\
    q_0 = 1,   &\qquad q_1 = a_1,         &\qquad q_k = a_k q_{k-1} + q_{k-2}, \qquad& k \geq 2.
  \end{alignat*}  
  For convenience, we set $p_{-1} = 1$ and $q_{-1} = 0$. The \emph{semiconvergents} (or intermediate fractions)
  $p_{k,\ell} / q_{k,\ell}$ of $\alpha$ are defined as the fractions
  \begin{equation*}
    \frac{ \ell p_{k-1} + p_{k-2} }{ \ell q_{k-1} + q_{k-2} }
  \end{equation*}
  for $1 \leq \ell < a_k$ and $k \geq 2$ (if they exist). Notice that semiconvergents are not a subtype of convergents.
  We often do not refer to convergents or semiconvergents, but to their denominators $q_k$ or $q_{k,\ell}$, so we let
  $\qk{\alpha}$ denote the set of denominators of convergents of $\alpha$ and $\qkl{\alpha}$ denote the set of
  denominators of the convergents and semiconvergents of $\alpha$. We emphasize that the above number $q_{-1}$ defined
  for convenience does not belong to the sets $\qk{\alpha}$ and $\qkl{\alpha}$. Throughout the paper, we make the
  convention that $\alpha$ always refers to some fixed irrational number in $(0, 1)$ with continued fraction expansion
  \eqref{eq:cf}, convergents $q_k$, and semiconvergents $q_{k,\ell}$.

  \begin{example}
    Let $\varphi$ be the Golden ratio, that is, set $\varphi = (1 + \sqrt{5})/2$. Then
    $\varphi = [1; \overline{1}] \approx 1.62$. The number $1/\varphi^2$, approximately $0.38$, has continued fraction
    expansion $[0; 2, \overline{1}]$. Its convergents, related to the Fibonacci numbers, are
    \begin{equation*}
      \frac{0}{1}, \frac{1}{2}, \frac{1}{3}, \frac{2}{5}, \frac{3}{8}, \frac{5}{13}, \ldots.
    \end{equation*}
    Notice that this number does not have semiconvergents.
  \end{example}

  For a real number $x$, we let $\{x\}$ to be its fractional part and $\norm{x} = \min\{\{x\}, 1 - \{x\}\}$. Here
  $\norm{x}$ measures the distance of $x$ to the nearest integer. It is often useful to reduce numbers of the form
  $n\alpha$ modulo $1$ and imagine them lying on the circle $\T$ having circumference $1$, which we identify with the
  unit interval $[0, 1)$. See \autoref{fig:points} for a picture of the numbers $\{-\alpha\}$, $\{-2\alpha\}$,
  $\ldots$, $\{-5\alpha\}$ lying on $\T$ when $\alpha = 1/\varphi^2$. In fact, adding $\alpha$ to its multiple can be
  viewed as the rotation
  \begin{equation*}
    R\colon \T \to \T, R(x) = \{x + \alpha\}
  \end{equation*}
  on $\T$.

  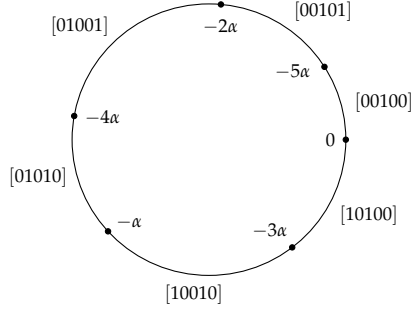
\begin{figure}
   \centering
   \begin{tikzpicture}
     \newcommand\CircleRadiusI{1.8}
     \newcommand\Slope{0.381966}

     \draw (0,0) circle (\CircleRadiusI);
     \foreach \n in {0,...,5}{
       \pgfmathparse{(-1)*\n*\Slope}
       \filldraw[fill=black] ({(\pgfmathresult - floor(\pgfmathresult)) * 360}:\CircleRadiusI) circle (1pt);
     }

     \node at (2.3,0.5)    {\scriptsize{$[00100]$}};
     \node at (1.5,1.7)    {\scriptsize{$[00101]$}};
     \node at (-1.7,1.5)   {\scriptsize{$[01001]$}};
     \node at (-2.25,-0.5) {\scriptsize{$[01010]$}};
     \node at (-0.2,-2.05) {\scriptsize{$[10010]$}};
     \node at (2.1,-1.0)   {\scriptsize{$[10100]$}};

     \node at (1.6,0) {\scriptsize{$0$}};
     \node at (-1.05,-1.1) {\scriptsize{$-\alpha$}};
     \node at (0.15,1.5) {\scriptsize{$-2\alpha$}};
     \node at (0.8,-1.2) {\scriptsize{$-3\alpha$}};
     \node at (-1.4,0.3) {\scriptsize{$-4\alpha$}};
     \node at (1.1,0.9) {\scriptsize{$-5\alpha$}};
   \end{tikzpicture}
   \caption{The points $0$, $\{-\alpha\}$, $\{-2\alpha\}$, $\ldots$, $\{-5\alpha\}$ on the circle $\T$ when
            $\alpha = 1/\varphi^2$. The intervals of the factors of length $5$ of the Fibonacci word are also
            included.}
   \label{fig:points}
  \end{figure}
 
  The denominators of convergents of $\alpha$ satisfy the \emph{best approximation property}:
  \begin{equation*}
    \norm{q_k \alpha} = \min_{0 < n < q_{k+1}} \norm{n\alpha}.
  \end{equation*}
  This means that the point $\{q_k\alpha\}$ is closer to the point $0$ on $\T$ than the points $\{\alpha\}$,
  $\{2\alpha\}$, $\ldots$, $\{(q_{k+1} - 1)\alpha\}$. Information on the quality of approximation of the numbers
  $\{q_{k,\ell}\alpha\}$ related to semiconvergents is given in
  \cite[Prop.~2.2]{2015:characterization_of_repetitions_in_sturmian_words_a_new}, but this information is not needed in
  this paper. For deeper understanding how the special points $\{q_k\alpha\}$ and $\{q_{k,\ell}\alpha\}$ lie on $\T$
  see \autoref{fig:intervals} (ignore the negative signs for now; they are needed when we work with Sturmian words).
  The details on why the picture is correct are found in the proof of \autoref{lem:larger_lower_bound}. It is important
  to understand how the next point closest to $0$ is formed from the previously closest points $\{q_k\alpha\}$ and
  $\{q_{k-1}\alpha\}$. Notice that $q_{k+1,1} = q_k + q_{k-1}$. The point $\{q_{k+1,1}\alpha\}$ related to the
  denominator of the convergent or semiconvergent $q_{k+1,1}$ is formed by performing $q_k$ rotations on the point
  $\{q_{k-1}\alpha\}$. This point $\{q_{k+1,1}\alpha\}$ is closer to $0$ than $\{q_{k-1}\alpha\}$---as is evident from
  \autoref{fig:intervals}---but it is not necessarily closer than $\{q_k\alpha\}$ if $a_{k+1} > 1$. In fact, we have
  \begin{equation*}
    \norm{q_{k+1,1}\alpha} = \norm{q_{k-1}\alpha} - \norm{q_k\alpha}.
  \end{equation*}
  Then successive $q_k$ rotations are added forming the points $\{q_{k+1,2}\alpha\}$, $\ldots$,
  $\{q_{k+1,a_{k+1}-1}\alpha\}$ that are successively closer to $0$ than $\{q_{k-1}\alpha\}$, but not closer than
  $\{q_k\alpha\}$. Finally the point $\{q_{k+1,a_{k+1}}\alpha\}$, i.e., the point $\{q_{k+1}\alpha\}$, is closer to $0$
  than $\{q_k\alpha\}$.

  By the preceding description, we see that
  $\norm{q_{k,\ell}\alpha} = \norm{q_{k,\ell-1}\alpha} - \norm{q_{k-1}\alpha}$. From this identity, it is not
  difficult to derive by induction that
  \begin{equation}\label{eq:cf_identity}
    \alpha = p_{k+1}\norm{q_k\alpha} + \norm{q_{k+1}\alpha}
  \end{equation}
  for all $k \geq 1$ when $a_0 = 0$ (i.e., when $\alpha \in (0, 1)$). Let then $\alpha_t$ for $t \geq 1$ denote the
  number with the continued fraction expansion $[a_t; a_{t+1}, a_{t+2}, \ldots]$. A short proof by induction shows that
  \begin{equation}\label{eq:cf_identity_2}
    \frac{\norm{q_k\alpha}}{\norm{q_{k+1}\alpha}} = \alpha_{k+2}
  \end{equation}
  for all $k \geq 1$. The following identity is well-known (see, e.g.,
  \cite[Sect.~10.7]{1979:an_introduction_to_the_theory_of_numbers}) for $k \geq 0$:
  \begin{equation*}
    \alpha - \frac{p_k}{q_k} = \frac{(-1)^k}{q_k(\alpha_{k+1}q_k + q_{k-1})}.
  \end{equation*}
  This identity shows that
  \begin{equation}\label{eq:norm_alpha}
    \norm{q_k\alpha} = \frac{1}{\alpha_{k+1}q_k + q_{k-1}}
  \end{equation}
  for $k \geq 0$.

  We conclude by a simple lemma needed in \autoref{sec:abelian_exponents} and in the proof of \autoref{lem:cases_2_3}.

  \begin{lemma}\label{lem:consecutive_norm}
    Let $\ell$ be a nonnegative integer. Then
    \begin{equation*}
      \frac{1}{\norm{q_{k-1}\alpha} + \ell\norm{q_k \alpha}} = \frac{\alpha_{k+1}q_k + q_{k-1}}{\alpha_{k+1} + \ell}
    \end{equation*}
    for all $k \geq 1$.
  \end{lemma}
  \begin{proof}
    By \eqref{eq:norm_alpha}, we have for all $k \geq 1$ that
    \begin{align*}
      \norm{q_{k-1}\alpha} &= (\alpha_k q_{k-1} + q_{k-2})^{-1} = (q_k + [0; a_{k+1}, \ldots]q_{k-1})^{-1} \\
                           &= (q_k + (a_{k+1} + \alpha_{k+2}^{-1})^{-1}q_{k-1})^{-1} \\
                           &= ( (a_{k+1} + \alpha_{k+2}^{-1})^{-1}( (a_{k+1} + \alpha_{k+2}^{-1})q_k + q_{k-1}) )^{-1} \\
                           &= ( (a_{k+1} + \alpha_{k+2}^{-1})^{-1}(\alpha_{k+1}q_k + q_{k-1}) )^{-1}
    \end{align*}
    (the computation indeed works with the convention $q_{-1} = 0$ when $k = 1$), so
    \begin{equation*}
      \norm{q_{k-1}\alpha} + \ell \norm{q_k \alpha} = \frac{\ell + a_{k+1} + \alpha_{k+2}^{-1}}{\alpha_{k+1}q_k + q_{k-1}} = \frac{\ell + \alpha_{k+1}}{\alpha_{k+1}q_k + q_{k-1}}.
    \end{equation*}
  \end{proof}

  \subsection{Sturmian Words}
  For the purposes of this paper, Sturmian words are best defined as codings of orbits of irrational rotations on the
  circle $\T$. For alternative definitions and proofs of the facts listed below, we refer the reader to
  \cite{2002:algebraic_combinatorics_on_words,2002:substitutions_in_dynamics_arithmetics_and_combinatorics}.

  Let $\alpha \in (0, 1)$ be an irrational real number, and divide $\T$ into two disjoint intervals $I_0$ and $I_1$ by
  the points $0$ and $1 - \alpha$. The map $R\colon \T \to \T$, $R(\rho) = \{\rho + \alpha\}$ defines an irrational
  rotation on $\T$. We shall code the orbit of a point $\rho$ as follows. Let $\nu$ be the coding function
  \begin{equation*}
    \nu(x) = \begin{cases}
               0, &\text{if $x \in I_0$,} \\
               1, &\text{if $x \in I_1$,}
             \end{cases}
  \end{equation*}
  and let $\infw{s}_{\rho,\alpha}$ be the infinite binary word whose $n$th letter (indexing from $0$) equals
  $\nu(R^n(\rho))$. We call this infinite word $\infw{s}_{\rho,\alpha}$ a \emph{Sturmian word of slope $\alpha$ and
  intercept $\rho$}. Our definition leaves the behavior of $\nu$ on the endpoints of $I_0$ and $I_1$ ambiguous. To fix
  this, we dictate that there are exactly two options: either select $I_0 = [0, 1-\alpha)$ and $I_1 = [1-\alpha, 1)$
  (case $0 \in I_0$) or set $I_0 = (0, 1-\alpha]$ and $I_1 = (1-\alpha, 1]$ (case $0 \notin I_0$). For a typical
  intercept $\rho$, the choice makes no difference, but a difference is seen if $\rho$ is of the form $\{-n\alpha\}$
  for some $n \geq 0$. We define the \emph{Sturmian subshift of slope $\alpha$}, denoted by $\Omega_\alpha$, to be the
  set of all Sturmian words of slope $\alpha$ and intercept $\rho$ obtained in both cases $0 \in I_0$ and $0 \notin I_0$.
  We refer to the words of $\Omega_\alpha$ as the Sturmian words of slope $\alpha$. We remark that
  $\Omega_\alpha \cap \Omega_\beta \neq \emptyset$ if and only if $\alpha = \beta$.

  \begin{example}
    Let $\alpha = 1/\varphi^2$ where $\varphi$ is the Golden ratio. Here $\alpha = [0; 2, \overline{1}] \approx 0.38$.
    The Sturmian word
    \begin{equation*}
      \infw{s}_{\alpha,\alpha} = 010010100100101001010010010100100101001010010010100101001001010010 \dotsm
    \end{equation*}
    of slope $\alpha$ and intercept $\alpha$ is the Fibonacci word $\infw{f}$ mentioned in the introduction. The
    subshift $\Omega_\alpha$ of slope $\alpha$ is called the \emph{Fibonacci subshift}.
  \end{example}

  The decision if $0 \in I_0$ or $0 \notin I_0$ is often irrelevant because all words in $\Omega_\alpha$ have the same
  \emph{language} (set of factors) $\Lang{\alpha}$. It is irrelevant in this paper too with the exceptions of the
  proofs of two minor claims. However, both options are needed for proving equivalence with the alternative definitions
  of Sturmian words and are needed to make $\Omega_\alpha$ a subshift.

  Let $w$ be a word $a_0 a_1 \dotsm a_{n-1}$ of length $n$ in $\Lang{\alpha}$, and set
  \begin{equation*}
    [w] = I_{a_0} \cap R^{-1}(I_{a_1}) \cap \dotsm \cap R^{-(n-1)}(I_{a_{n-1}}).
  \end{equation*}
  Then $[w]$ is the unique subinterval of $\T$ such that $\infw{s}_{\rho,\alpha}$ begins with $w$ if and only if
  $\rho \in [w]$. The points $0$, $\{-\alpha\}$, $\ldots$, $\{-n\alpha\}$ partition the circle $\T$ into $n + 1$
  subintervals that are in one-to-one correspondence with the words of $\Lang{\alpha}$ of length $n$. See
  \autoref{fig:points} for the intervals of the factors of length $5$ of the Fibonacci word. We let $I(x, y)$,
  $\{x\} < \{y\}$, stand for the interval $[\{x\}, \{y\})$ if $0 \in I_0$ and for $(\{x\}, \{y\}]$ if $0 \notin I_0$.
  We call the words of $\Lang{\alpha}$ the \emph{factors of slope $\alpha$}.

  Moreover, Sturmian words are recurrent, and the language $\Lang{\alpha}$ is \emph{closed under reversal}: for each
  word $w$ in $\Lang{\alpha}$, its reversal $\mirror{w}$ is also in $\Lang{\alpha}$. The only difference between
  Sturmian words of slope $[0; 1, a_2, a_3, \ldots]$ and Sturmian words of slope $[0; a_2 + 1, a_3, \ldots]$ is that
  the roles of the letters $0$ and $1$ are reversed. Thus we make the typical assumption that $a_1 \geq 2$ in
  \eqref{eq:cf}. This means that $\alpha \in (0, \tfrac12)$.

  \subsection{Abelian Powers, Repetitions, and Periods}
  Many of the notions and results presented in this subsection and the following subsection are found in
  \cite{2016:abelian_powers_and_repetitions_in_sturmian_words}. However, we use the notation of
  \cite[Ch.~4.7]{diss:jarkko_peltomaki}.

  Let $w$ be a finite binary word over the alphabet $\{0, 1\}$. The \emph{Parikh vector} (or abelianization)
  $\Parikh{w}$ of $w$ is defined to be the vector $(\abs{w}_0, \abs{w}_1)$ counting the number of occurrences of the
  letters $0$ and $1$ in $w$. Two words $u$ and $v$ are \emph{abelian equivalent} if $\Parikh{u} = \Parikh{v}$. If
  $\mathcal{P}$ and $\mathcal{Q}$ are two Parikh vectors and $\mathcal{P}$ is componentwise less than or equal to
  $\mathcal{Q}$ but is not equal to $\mathcal{Q}$, then we say that $\mathcal{P}$ is \emph{contained} in $\mathcal{Q}$.

  Using the above notions, we generalize the notion of a period to the abelian setting.

  \begin{definition}\label{def:abelian_period}
    An \emph{abelian decomposition} of a word $w$ is a factorization $w = u_0 u_1 \dotsm u_{n-1} u_n$ such that
    $n \geq 2$, the words $u_1$, $\ldots$, $u_{n-1}$ have a common Parikh vector $\mathcal{P}$ (i.e., they are abelian
    equivalent), and the Parikh vectors of $u_0$ and $u_n$ are contained in $\mathcal{P}$. The words $u_0$ and $u_n$
    are respectively called the \emph{head} and the \emph{tail} of the decomposition. The common length $m$ of the
    words $u_1$, $\ldots$, $u_{n-1}$ is called an \emph{abelian period} of $w$. The \emph{minimum abelian period}
    (i.e., the shortest) of $w$ is denoted by $\mu_w$.

    If $n \geq 3$, then we say that $w$ is an \emph{abelian repetition} of period $m$ and exponent $\abs{w}/m$. If
    $n \geq 3$ and the head $u_0$ and the tail $u_n$ are empty, then we say that $w$ is an \emph{abelian power} of
    period $m$ and exponent $\abs{w}/m$. If $n \leq 2$, then we say that $w$ is a \emph{degenerate abelian repetition}
    (of period $m$) or a \emph{degenerate abelian power} (of period $m$) if the head and tail are empty.
  \end{definition}

  For example, the word $abaababa$ has abelian decompositions $a \cdot ba \cdot ab \cdot ab \cdot a$ (of period $2$ and
  exponent $8/2$) and $\varepsilon \cdot aba \cdot aba \cdot ba$ (of period $3$ and exponent $8/3$).

  The following lemma is immediate.

  \begin{lemma}\label{lem:abelian_period_covering}
    Let $u$ be a factor of a word $w$. Then $\mu_w \geq \mu_u$. On the other hand, if $w$ has an abelian period $m$
    such that $m \leq \abs{u}$, then $m$ is also an abelian period of $u$.
  \end{lemma}

  \begin{definition}
    Let $w$ be a finite or infinite word. Then the set
    \begin{equation*}
      \{\mu_u : \text{$u$ is a nonempty factor of $w$}\}
    \end{equation*}
    is called the \emph{abelian period set of $w$}.
  \end{definition}

  \subsection{Abelian Powers in Sturmian Words}
  The starting point of the study of abelian equivalence in Sturmian words is the following result stating that factors
  of length $n$ of a Sturmian word belong to exactly two abelian equivalence classes and that these classes can be
  identified with the subintervals of $\T$ separated by the points $0$ and $\{-n\alpha\}$. Let $w$ be a factor of slope
  $\alpha$. If $w$ contains the minimum (resp. maximum) number of occurrences of the letter $1$ among factors of length
  $\abs{w}$, then we say that $w$ is \emph{light} (resp. \emph{heavy}).

  \begin{proposition}\label{prp:ab_eq_geometric}
    \cite[Prop.~3.3]{2016:abelian_powers_and_repetitions_in_sturmian_words},
    \cite[Thm.~19]{2013:some_properties_of_abelian_return_words}
    Each factor of length $n$ in $\Lang{\alpha}$ is either light or heavy. A factor $w$ in $\Lang{\alpha}$ is light if
    and only if $[w] \subseteq I(0, -\abs{w}\alpha\}$. Moreover, if $\{-n\alpha\} \geq 1 - \alpha$, then all heavy
    factors of length $n$ begin and end with $1$, while if $\{-n\alpha\} \leq 1 - \alpha$, then each light factor of
    length $n$ begins and ends with $0$.
  \end{proposition}

  The following proposition is a direct consequence of \autoref{prp:ab_eq_geometric}, but it is best to state it for
  clarity. See \cite[Lemma~4.2]{2016:abelian_powers_and_repetitions_in_sturmian_words} for more precise information.

  \begin{proposition}\label{prp:ab_eq_geometric_2}
    Let $\infw{s}_{\rho,\alpha} = a_0 a_1 \dotsm$ be a Sturmian word of slope $\alpha$ and intercept $\rho$. Then its
    factor $a_n \dotsm a_{n+m-1} \dotsm a_{n+em-1}$ is an abelian power of period $m$ and exponent $e$, $e \geq 2$, if
    and only if the $e$ points $\{\rho + (n + im)\alpha\}$, $i = 0$, $\ldots$, $e - 1$, are all either in the interval
    $I(0, -m\alpha)$ or in the interval $I(-m\alpha, 1)$.
  \end{proposition}

  \begin{remark}\label{rem:singular}
    Consider factors of a Sturmian word of slope $\alpha$ and of length $q_k$ for some $k \geq 0$. By the best
    approximation property, the point $\{-q_k\alpha\}$ is closest to the point $0$ among the points $\{-\alpha\}$,
    $\{-2\alpha\}$, $\ldots$, $\{-q_k \alpha\}$. This means that the interval separated by the points $\{-q_k\alpha\}$
    and $0$ is the interval $[s]$ of a unique word $s$ of length $q_k$. The word $s$ is called the \emph{singular
    factor} of length $q_k$. By \autoref{prp:ab_eq_geometric}, we see that the factors of length $q_k$ that do not
    equal $s$ are abelian equivalent.
    
    The singular factors play a crucial role in deriving the main inequality in \autoref{sec:main_inequality}. Singular
    factors have been studied before in other contexts; see
    \cite{1999:lyndon_words_and_singular_factors_of_sturmian,2003:some_properties_of_the_factors_of_sturmian_sequences}.
    The previous approaches have been combinatorial, but here we derive the needed results by number-theoretic means.
  \end{remark}

  We need the following result on singular factors.

  \begin{lemma}\label{lem:singular_properties}
    The singular factor $s$ of length $q_k$ has the following properties:
    \begin{enumerate}[(i)]
      \item $s$ begins and ends with the same letter;
      \item $s$ is a palindrome; and
      \item the Parikh vectors of proper prefixes and suffixes of $s$ are contained in the Parikh vectors of all
            factors of length $q_k$.
    \end{enumerate}
  \end{lemma}
  \begin{proof}
    The property (i) is directly implied by \autoref{prp:ab_eq_geometric}. Namely if $\{-q_k\alpha\} \geq 1 - \alpha$,
    then $s$ is heavy by \autoref{rem:singular} and $s$ begins and ends with $1$ by \autoref{prp:ab_eq_geometric}. If
    $\{-q_k\alpha\} \leq 1 - \alpha$, then $s$ is light and begins and ends with $0$.

    Since the language $\Lang{\alpha}$ is closed under reversal, we have $\mirror{s} \in \Lang{\alpha}$ (recall that
    $\mirror{s}$ is the reversal of $s$). By \autoref{rem:singular}, the singular factor $s$ uniquely corresponds to
    its Parikh vector among factors of length $q_k$. Since a Parikh vector is invariant under reversal, it follows that
    $\mirror{s} = s$. This establishes property (ii).

    Let us then consider the final claim. If $\abs{s} = 1$, then there is nothing to prove, so suppose that
    $\abs{s} > 1$. Write $s = as' a$ for a letter $a$. It is sufficient to prove that the Parikh vector $\Parikh{as'}$
    of $as'$ is contained in the Parikh vectors of all factors of length $q_k$ because $\Parikh{s'a} = \Parikh{as'}$.
    Let $w$ be a factor of length $q_k$ such that $w \neq s$. Suppose that $a = 0$. The first paragraph of this proof
    shows that $s$ is light. This means that $w$ is heavy. Thus $\abs{w}_1 > \abs{s}_1 = \abs{as'}_1$. In fact, since
    all factors of fixed length are either heavy or light by \autoref{prp:ab_eq_geometric}, it must be that
    $\abs{w}_1 = \abs{s}_1 + 1$. In other words, $\abs{w}_0 = \abs{s}_0 - 1 = \abs{as'}_0$. Hence $\Parikh{as'}$ is
    contained in $\Parikh{w}$. The case $a = 1$ is similar.
  \end{proof}

  \autoref{prp:ab_eq_geometric} allows continued fractions and geometric arguments to be applied to the study of
  abelian powers in Sturmian words. Let $\abexp{m}$ denote the maximum exponent of an abelian power of period $m$
  occurring in a Sturmian word of slope $\alpha$. The number $\abexp{m}$ is always finite and is easily computed using
  the following result.

  \begin{proposition}\label{prp:exp_formula} \cite[Thm.~4.7]{2016:abelian_powers_and_repetitions_in_sturmian_words}
    We have $\abexp{m} = \Floor{\tfrac{1}{\norm{m\alpha}}}$.
  \end{proposition}
  \begin{proof}[Proof Sketch]
    We sketch the proof here because similar arguments are needed in the proofs of \autoref{lem:first_last_letter} and
    \autoref{lem:cases_2_3}. Say $\{-m\alpha\} < \tfrac12$; the case $\{-m\alpha\} > \tfrac12$ is similar. Consider two
    points $\rho$ and $\{\rho + m\alpha\}$ on $\T$. Because the distance between these points is $\norm{m\alpha}$, they
    cannot both belong to the interval $I(0, -m\alpha)$. If they lie on the interval $I(-m\alpha, 1)$ of length
    $1 - \norm{m\alpha}$, then the word $\infw{s}_{\rho,\alpha}$ begins with an abelian square of period $m$ by
    \autoref{prp:ab_eq_geometric_2}. To find the maximum exponent of an abelian power of period $m$ that is a prefix of
    $\infw{s}_{\rho,\alpha}$, it thus suffices to see how many times $\norm{m\alpha}$ divides $1 - \norm{m\alpha}$.
    This proves the claim.
  \end{proof}

  Since $\norm{m\alpha}$ can be made as small as desired, the preceding proposition shows that each Sturmian word
  contains abelian powers of arbitrarily high exponent. A similar result for a broader class of words is given in
  \cite[Thm.~1.8]{2011:abelian_complexity_of_minimal_subshifts}.

  \section{Lemmas on Abelian Exponents}\label{sec:abelian_exponents}
  In this section, we prove several inequalities concerning the abelian exponents of factors of slope $\alpha$ needed
  mainly in \autoref{sec:main_proofs}. The results presented here are purely arithmetical in their nature and do not,
  as such, provide any significant insight for proving the main results. The reader might want to read
  \autoref{sec:main_inequality} before studying this section in detail.

  The first lemma relates an abelian exponent to a convergent of $\alpha$.

  \begin{lemma}\label{lem:relation_exponent_convergent}
    If $\norm{m\alpha} \geq \norm{q_{k-1}\alpha} + \norm{q_k\alpha}$ for some $k \geq 1$, then $\abexp{m} < q_k$.
  \end{lemma}
  \begin{proof}
    Suppose that $\norm{m\alpha} \geq \norm{q_{k-1}\alpha} + \norm{q_k\alpha}$ for some $k \geq 1$. By
    \autoref{prp:exp_formula}, it suffices to establish that
    \begin{equation*}
      \frac{1}{\norm{q_{k-1}\alpha} + \norm{q_k \alpha}} < q_k.
    \end{equation*}
    By \autoref{lem:consecutive_norm}, this inequality is equivalent to the inequality
    \begin{equation*}
      \alpha_{k+1}q_k + q_{k-1} < \alpha_{k+1}q_k + q_k,
    \end{equation*}
    which is obviously true because $q_{k-1} < q_k$.
  \end{proof}

  \autoref{lem:relation_exponent_convergent} is sharp in the sense that, for suitable partial quotients, it is possible
  that $\abexp{m} = q_k - 1$ for some $k$. For instance, if $\alpha = [0; 2, \overline{1}]$ and $m = 4$, then
  $q_2 = 3 < m < q_3 = 5$, $\norm{m\alpha} \approx 0.47 > 0.38 \approx \norm{q_1\alpha} + \norm{q_2\alpha}$, and
  $\abexp{m} = 2 = q_2 - 1$.

  \begin{figure}
  \centering
  \begin{tikzpicture}
    \tikzstyle{own}=[line width=0.8pt,black]
    \tikzstyle{point}=[line width=0.5pt,black,fill=white]
    \tikzstyle{point2}=[line width=0.8pt,black]

    \tikzmath{\width = 14.5; \theight = 0.5; \lheight = -5; \buheight = 1.0; \blheight = -1.0; \q1 = 0.07; \q2 = 0.16; \q3 = \q1 + 2*\q2; \q4 = \q2 + 2*\q3;}

    \draw[own,-|] (0,0) -- (\width,0);
    \node at (\width,\theight) {\footnotesize $0$};
    \draw[own,|-] (0,\lheight) -- (\width,\lheight);
    \pgfmathsetmacro\result{\lheight + \theight}
    \node at (0,\result) {\footnotesize $0$};

    \pgfmathsetmacro\result{(1 - \q1)*\width}
    \draw[point] ({\result}, 0) circle(2.5pt);
    \node at ({\result}, \theight) {\footnotesize $-q_{k+1}\alpha$};
    \pgfmathsetmacro\result{(1 - \q3)*\width}
    \filldraw[point] ({\result}, 0) circle(2.5pt);
    \node at ({\result}, \theight) {\footnotesize $-q_{k-1}\alpha$};
    \pgfmathsetmacro\result{(1 - 2*\q3)*\width}
    \filldraw[point] ({\result}, 0) circle(2.5pt);
    \node at ({\result}, \theight) {\footnotesize $-2q_{k-1}\alpha$};
    \pgfmathsetmacro\result{(1 - \q1 - \q2 - \q3)*\width}
    \filldraw[point2] ({\result}, 0) circle(0.9pt);
    \node at ({\result}, -\theight) {\footnotesize $-(2q_{k-1} + (a_{k+1} - 1)q_k)\alpha$};

    \pgfmathsetmacro\result{\q2*\width}
    \filldraw[point] ({\result}, \lheight) circle(2.5pt);
    \node at ({\result}, \lheight+\theight) {\footnotesize $-q_k\alpha$};
    \pgfmathsetmacro\result{\q4*\width}
    \filldraw[point] ({\result}, \lheight) circle(2.5pt);
    \node at ({\result}, \lheight+\theight) {\footnotesize $-q_{k-2}\alpha$};
    \pgfmathsetmacro\result{(\q2 + \q3)*\width}
    \filldraw[point] ({\result}, \lheight) circle(2.5pt);
    \node at ({\result}, \lheight+\theight) {\footnotesize $-q_{k,a_k-1}\alpha$};
    \pgfmathsetmacro\result{2*\q2*\width}
    \filldraw[point2] ({\result}, \lheight) circle(0.9pt);
    \node at ({\result}, \lheight-\theight) {\footnotesize $-2q_k\alpha$};
    \pgfmathsetmacro\result{(\q2+\q3+2*\q2)*\width}
    \filldraw[point2] ({\result}, \lheight) circle(0.9pt);
    \node at ({\result}, \lheight-\theight) {\footnotesize $-(q_{k,a_k-1} + a_{k+1}q_k)\alpha$};

    \pgfmathsetmacro\result{(1 - \q3)*\width}
    \pgfmathsetmacro\tresult{(1 - \q1)*\width}
    \draw [thick,decoration={brace},decorate] (\result,\buheight) -- (\tresult,\buheight) node [midway,align=center,yshift=20] {\footnotesize $a_{k+1}\norm{q_k\alpha}$};
    \pgfmathsetmacro\result{(1 - \q1)*\width}
    \pgfmathsetmacro\tresult{\width}
    \draw [thick,decoration={brace},decorate] (\result,\buheight) -- (\tresult,\buheight) node [midway,align=center,yshift=20] {\footnotesize $\norm{q_{k+1}\alpha}$};
    \pgfmathsetmacro\result{(1 - \q1 - \q2 - \q3)*\width}
    \pgfmathsetmacro\tresult{(1 - \q3)*\width}
    \draw [thick,decoration={brace},decorate] (\result,\buheight) -- (\tresult,\buheight) node [midway,align=center,yshift=20] {\footnotesize $\norm{q_k\alpha} + \norm{q_{k+1}\alpha}$};

    \pgfmathsetmacro\result{(1 - \q1)*\width}
    \pgfmathsetmacro\tresult{(1 - \q3)*\width}
    \draw [thick,decoration={brace,mirror},decorate] (\tresult,\blheight) -- (\result,\blheight) node [midway,align=center,yshift=-20] {\footnotesize (semi)convergents\\\footnotesize $-q_{k+1,\ell}\alpha, \, 0 \leq \ell \leq a_{k+1}$};
    \pgfmathsetmacro\result{(1 - \q1 - \q2 - \q3)*\width}
    \pgfmathsetmacro\tresult{(1 - 2*\q3)*\width}
    \draw [thick,decoration={brace,mirror},decorate] (\tresult,\blheight) -- (\result,\blheight) node
  [midway,align=center,yshift=-20] {\footnotesize points\\\footnotesize $-(2q_{k-1} + \ell q_k)\alpha, \, 0 \leq \ell < a_{k+1}$};

    \pgfmathsetmacro\result{0*\width}
    \pgfmathsetmacro\tresult{\q2*\width}
    \draw [thick,decoration={brace},decorate] (\result,\lheight+\buheight) -- (\tresult,\lheight+\buheight) node [midway,align=center,yshift=20] {\footnotesize $\norm{q_k\alpha}$};
    \pgfmathsetmacro\result{\q2*\width}
    \pgfmathsetmacro\tresult{(\q2+\q3)*\width}
    \draw [thick,decoration={brace},decorate] (\result,\lheight+\buheight) -- (\tresult,\lheight+\buheight) node [midway,align=center,yshift=20] {\footnotesize $\norm{q_{k-1}\alpha}$};
    \pgfmathsetmacro\result{(\q2+\q3)*\width}
    \pgfmathsetmacro\tresult{(\q2+\q3+2*\q2)*\width}
    \draw [thick,decoration={brace},decorate] (\result,\lheight+\buheight) -- (\tresult,\lheight+\buheight) node [midway,align=center,yshift=20] {\footnotesize $a_{k+1}\norm{q_k\alpha}$};

    \pgfmathsetmacro\result{2*\q2*\width}
    \pgfmathsetmacro\tresult{(\q2 + \q3)*\width - 0.3}
    \draw [thick,decoration={brace,mirror},decorate] (\result,\lheight+\blheight) -- (\tresult,\lheight+\blheight) node [midway,align=center,yshift=-20] {\footnotesize multiples\\\footnotesize $-tq_k\alpha, \, 1 < t \leq a_{k+1}$};
    \pgfmathsetmacro\result{(\q2 + \q3)*\width}
    \pgfmathsetmacro\tresult{\q4*\width}
    \draw [thick,decoration={brace,mirror},decorate] (\result,\lheight+\blheight) -- (\tresult,\lheight+\blheight) node [midway,align=center,yshift=-33] {\footnotesize semiconvergents\\\footnotesize $-q_{k,\ell}\alpha, \, 0 \leq \ell < a_{k+1}$,\\\footnotesize points\\\footnotesize $-(q_{k,\ell} + t q_k)\alpha, \, 1 \leq \ell < a_k, 0 \leq t < a_{k+1}$};
  \end{tikzpicture}
  \caption{The points $\{-i\alpha\}$ with $i \leq q_{k+1}$ that are closest to $0$. The picture is in scale;
  $a_k = a_{k+1} = 2$ was used for drawing.}
  \label{fig:intervals}
  \end{figure}
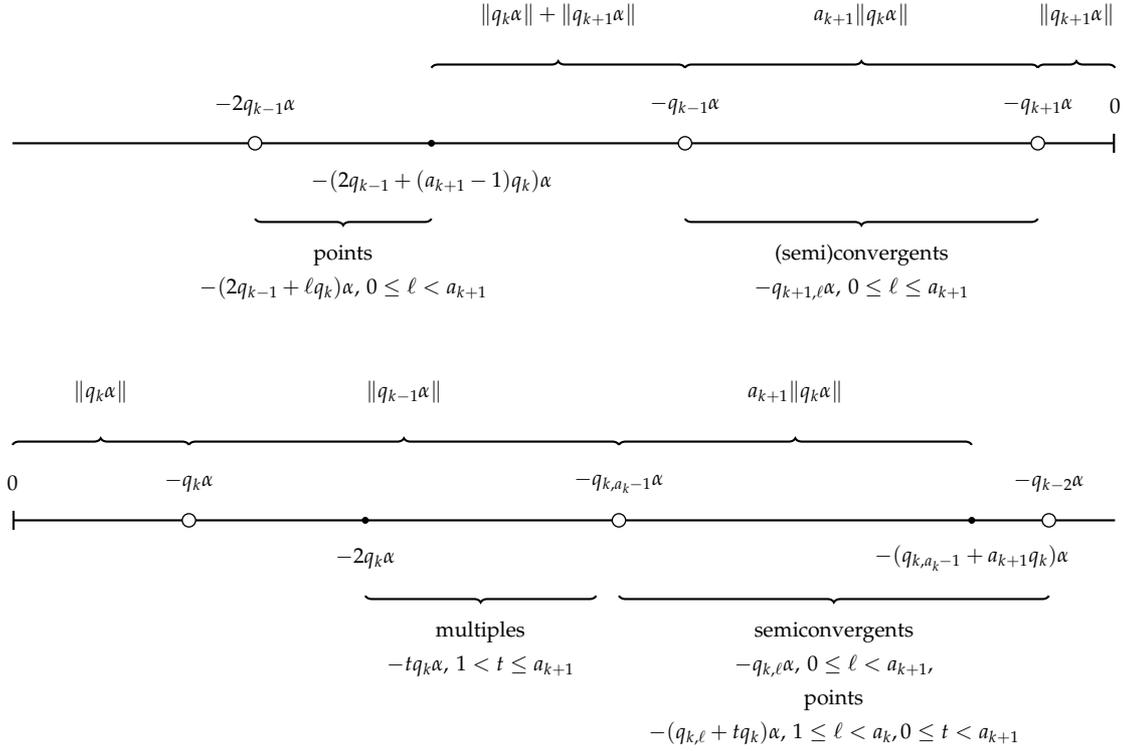

  In some cases, we need the following improvement of \autoref{lem:relation_exponent_convergent}.

  \begin{lemma}\label{lem:relation_exponent_convergent_2}
    If $\norm{m\alpha} \geq \norm{q_{k-1}\alpha} + (a_{k+1} + 1)\norm{q_k\alpha}$ for some $k \geq 2$, then
    $\abexp{m} < q_k - 1$.
  \end{lemma}
  \begin{proof}
    Suppose that $\norm{m\alpha} \geq \norm{q_{k-1}\alpha} + (a_{k+1} + 1)\norm{q_k\alpha}$ for some $k \geq 2$. Then,
    by \autoref{lem:consecutive_norm}, it suffices to show that
    \begin{equation*}
      \frac{\alpha_{k+1}q_k + q_{k-1}}{\alpha_{k+1} + a_{k+1} + 1} < q_k - 1.
    \end{equation*}
    This inequality is equivalent to
    \begin{equation}\label{eq:foo3}
      q_{k-1} < (a_{k+1} + 1)q_k - (\alpha_{k+1} + a_{k+1} + 1).
    \end{equation}
    Now $q_k \geq q_{k-1} + q_{k-2}$, so it is enough to show that
    \begin{equation}\label{eq:sufficient_2}
      a_{k+1}q_{k-1} + (a_{k+1} + 1)q_{k-2} > \alpha_{k+1} + a_{k+1} + 1.
    \end{equation}
    Since $k \geq 2$, we have $q_{k-1} \geq q_1 \geq 2$ and $q_{k-2} \geq q_0 = 1$. Thus
    $a_{k+1}q_{k-1} + (a_{k+1} + 1)q_{k-2} \geq 3a_{k+1} + 1$ and $3a_{k+1} + 1 > \alpha_{k+1} + a_{k+1} + 1$ if and
    only if $2a_{k+1} > \alpha_{k+1}$ (recall that $\alpha_{k+1} < a_{k+1} + 1$). Since $a_{k+1} \geq 1$, this final
    inequality is true. This means that \eqref{eq:sufficient_2} holds.
  \end{proof}

  In order to apply \autoref{lem:relation_exponent_convergent_2}, we need the following lemma. Its proof essentially
  argues that \autoref{fig:intervals} is correctly drawn. This figure is important for the proofs in
  \autoref{sec:main_proofs}. It depicts the points of the form $\{-i\alpha\}$ with $i \leq q_{k+1}$ that are closest to
  $0$. The presented arguments contain ingredients for proving the Three Distance Theorem; see
  \cite{1998:three_distance_theorems_and_combinatorics_on_words} and its references, especially
  \cite{1988:the_three_gap_theorem_steinhaus_conjecture}.

  \begin{lemma}\label{lem:larger_lower_bound}
    Let $k \geq 1$, and suppose that $m$ is an integer such that $a_{k+1}q_k < m < q_{k+1}$. If
    $m \neq q_{k+1,a_{k+1}-1}$ and $m \neq (a_{k+1}-1)q_k + 2q_{k-1}$ when $a_k = 1$, then
    $\norm{m\alpha} \geq \norm{q_{k-1}\alpha} + (a_{k+1} + 1)\norm{q_k\alpha}$.
  \end{lemma}
  \begin{proof}
    Suppose that $m$ does not equal $q_{k+1,a_{k+1}-1}$, and assume moreover that if $a_k = 1$, then
    $m \neq (a_{k+1}-1)q_k + 2q_{k-1}$. For the proof, we omit the negative signs and consider points of the form
    $\{i\alpha\}$ with $i$ positive instead the points $\{-i\alpha\}$ that are the endpoints of the intervals of the
    factors of slope $\alpha$.

    Assume first that the point $\{m\alpha\}$ is on the same side of the point $0$ as the point $\{q_{k+1}\alpha\}$.
    By this we mean that if $\{q_{k+1}\alpha\} < \tfrac12$, then also $\{m\alpha\} < \tfrac12$ and if
    $\{q_{k+1}\alpha\} > \tfrac12$, then $\{m\alpha\} > \tfrac 12$. By the best approximation property, the point
    $\{m\alpha\}$ cannot be closer to $0$ than $\{q_{k+1}\alpha\}$. Let $D_1$ be the distance of $\{m\alpha\}$ to $0$
    through the point $\{q_{k+1}\alpha\}$. Points $\{i\alpha\}$ with $i < q_{k+1}$ between the points
    $\{q_{k-1}\alpha\}$ and $\{q_{k+1}\alpha\}$ are exactly the points $\{q_{k+1,\ell}\alpha\}$ for
    $1 \leq \ell < a_{k+1}$ because of the best approximation property and the fact that the distance between
    $\{q_{k+1,\ell}\alpha\}$ and $\{q_{k+1,\ell+1}\alpha\}$ is $\norm{q_k\alpha}$. Since $a_{k+1}q_k < m$ and
    $m \neq q_{k+1,a_{k+1}-1}$, we conclude that $D_1 > \norm{q_{k-1}\alpha}$. Let us consider next points between
    $\{q_{k-1}\alpha\}$ and $\{2q_{k-1}\alpha\}$. The points $\{(2q_{k-1} + \ell q_k)\alpha\}$,
    $1 \leq \ell \leq a_{k+1}$, lie between $\{q_{k-1}\alpha\}$ and $\{2q_{k-1}\alpha\}$. As the distance between two
    consecutive such points is $\norm{q_k\alpha}$, the points between $\{q_{k-1}\alpha\}$ and $\{2q_{k-1}\alpha\}$ of
    the form $\{i\alpha\}$ with $i < q_{k+1}$ are among these points $\{(2q_{k-1} + \ell q_k)\alpha\}$. Say
    $m = 2q_{k-1} + \ell q_k$ for some $\ell$ such that $1 \leq \ell \leq a_{k+1}$. Then the assumption
    $a_{k+1} q_k < m < q_{k+1}$ implies that $\ell = a_{k+1} - 1$ and $q_k < 2q_{k-1}$. The inequality $q_k < 2q_{k-1}$
    implies that $a_k = 1$. This case is however excluded by our assumptions. Thus $\{m\alpha\}$ does not lie between
    $\{q_{k-1}\alpha\}$ and $\{2q_{k-1}\alpha\}$. Consider then the point $\{i\alpha\}$ with $i < q_{k+1}$ that is
    closest to the point $\{2q_{k-1}\alpha\}$. By the best approximation property, the distance from $\{i\alpha\}$ to
    $\{2q_{k-1}\alpha\}$ cannot be less than or equal to $\norm{q_{k+1}\alpha}$. Therefore it must be at least
    $\norm{q_k\alpha}$. Therefore (see \autoref{fig:intervals})
    \begin{align*}
      D_1 &\geq 2\norm{q_{k-1}\alpha} + \norm{q_k\alpha} \\
          &= \norm{q_{k-1}\alpha} + a_{k+1}\norm{q_k\alpha} + \norm{q_{k+1}\alpha} + \norm{q_k\alpha} \\
          &> \norm{q_{k-1}\alpha} + (a_{k+1} + 1)\norm{q_k\alpha}.
    \end{align*}

    Assume then that the point $\{m\alpha\}$ is on the same side as $\{q_k\alpha\}$. Again $\{m\alpha\}$ cannot be
    closer to $0$ than the point $\{q_k\alpha\}$ due to the best approximation property. Let $D_2$ be the distance of
    $\{m\alpha\}$ to $0$ through the point $\{q_k\alpha\}$. If $\{i\alpha\}$ with $i < q_{k+1}$ is a point between
    $\{q_k\alpha\}$ and $\{q_{k,a_k-1}\alpha\}$, then $i$ is a multiple of $q_k$. This means that $\{m\alpha\}$ is not
    between $\{q_k\alpha\}$ and $\{q_{k,a_k-1}\alpha\}$. The points closest to $\{q_{k,a_k-1}\alpha\}$ that are not
    between $\{q_k\alpha\}$ and $\{q_{k,a_k-1}\alpha\}$ are the points $\{(q_{k,a_k-1} + \ell q_k)\alpha\}$ with
    $1 \leq \ell \leq a_k$. If $q_{k,a_k-1} + \ell q_k > a_{k+1}q_k$, then $\ell = a_{k+1}$. Thus
    $D_2 \geq \norm{(q_{k,a_k-1} + a_{k+1} q_k)\alpha}$. The claim follows since
    \begin{equation*}
      \norm{(q_{k,a_k-1} + a_{k+1} q_k)\alpha} = a_{k+1}\norm{q_k\alpha} + \norm{q_{k,a_k-1}\alpha} = a_{k+1}\norm{q_k\alpha} + \norm{q_{k-1}\alpha} + \norm{q_k\alpha}
    \end{equation*}
    and $\norm{m\alpha} \geq \min\{D_1, D_2\}$.
  \end{proof}

  The next result contains a lower bound for an abelian exponent.

  \begin{lemma}\label{lem:exponent_lower_bound}
    If $\norm{m\alpha} \leq \norm{q_k\alpha}$, then $\abexp{m} \geq q_{k+1}$. If $\norm{m\alpha} \geq
    \norm{q_k\alpha}$, then $\abexp{m} < q_{k+1} + q_k$.
  \end{lemma}
  \begin{proof}
    The claim follows directly from \eqref{eq:norm_alpha}:
    $1/\norm{q_k\alpha} = \alpha_{k+1}q_k + q_{k-1} = q_{k+1} + [0; a_{k+2}, \ldots]q_k$ and
    $0 < [0; a_{k+2}, \ldots] < 1$.
  \end{proof}

  \section{Idea of the Proof and Derivation of the Main Inequality}\label{sec:main_inequality}
  The idea of the proof of \autoref{thm:fibonacci} that is given in
  \cite{2016:abelian_powers_and_repetitions_in_sturmian_words} is, roughly speaking, to show that near the beginning of
  an occurrence of a factor $w$ with abelian period $m$ in $\infw{f}$, the Fibonacci word, there begins an abelian
  power of period $F_k$, where $F_k$ is the largest Fibonacci number such that $F_k \leq m$, and large exponent that
  contains $w$ completely. This shows by \autoref{lem:abelian_period_covering} that $F_k$ is an abelian period of $w$.
  Thus the minimum abelian period of $w$ must be a Fibonacci number. Recall that all Sturmian words with a common slope
  have the same language. Therefore it often suffices to study the factors of a single Sturmian word of slope $\alpha$.

  As is mentioned in the introduction, an explicit counterexample showed that the above proof idea, as such, does not
  generalize to other Sturmian words. In fact, we shall show in \autoref{prp:fibonacci_counter_example} that such a
  counterexample exists in all cases except in the case of the slope $1/\varphi^2$. This means that the proof of
  \cite{2016:abelian_powers_and_repetitions_in_sturmian_words} for \autoref{thm:fibonacci} is specific to the Fibonacci
  subshift. While this specific proof could be modified to work more generally, this line of reasoning seems to be
  unworkable. Thus new ideas are necessary.

  Let us now consider Sturmian words of slope $\alpha$. Instead of looking for abelian powers of period $q_k$ with
  large exponent that can cover some factor of slope $\alpha$, the idea is to see what it means if the period $q_k$ is
  avoided. We shall soon see that abelian powers of period $q_k$ cover almost all of a Sturmian word of slope $\alpha$.
  This means that a factor $w$ avoiding the period $q_k$ must be rather long. This in turn means that the abelian
  exponent $\abexp{m}$, related to the minimum abelian period $m$ of $w$, must be large whenever $m$ is not too large
  compared to $q_k$. Since $\abexp{m} = \floor{1/\norm{m\alpha}}$ by \autoref{prp:exp_formula}, it must be that
  $\norm{m\alpha}$ is small. The analysis of \autoref{sec:main_proofs} indicates that $\norm{m\alpha}$ has to be so
  small that $m$ relates to a rather good rational approximation of $\alpha$. Precise analysis of the quality of the
  approximation leads to the statement of \autoref{thm:main_intro}.

  We let $\mathcal{M}_\alpha$ denote the set $\{tq_k : \text{$k \geq 0$ and $1 \leq t \leq a_{k+1}$}\}$. With the new
  notation from \autoref{sec:preliminaries}, we now rephrase the main result, \autoref{thm:main_intro}, as follows.
  
  \begin{theorem}\label{thm:main}
    If $m$ is the minimum abelian period of a nonempty factor of slope $\alpha$, then
    $m \in \qkl{\alpha} \cup \mathcal{M}_\alpha$.
  \end{theorem}

  Let us consider the minimum abelian period $m$ of a nonempty word $w$ in $\Lang{\alpha}$ (we shall use the notation
  introduced here throughout this section). In view of \autoref{thm:main}, we suppose that $m \notin
  \mathcal{M}_\alpha$. Let $k$ be the largest integer such that $q_k < m$, and let $t$ to be the largest integer such
  that $tq_k < m$ with $1 \leq t \leq a_{k+1}$. Notice that our assumptions imply that $k \geq 1$ because $q_0 = 1$. By
  taking the exponent and head and tail length to be maximal, we see that $\abs{w} \leq (\abexp{m} + 2)m - 2$. The main
  task of this section is to derive the following lower bound for the length of $w$:
  \begin{equation}\label{eq:main_lower_bound}
    (q_{k+1} + 2t - 1)q_k - q_{k+1} \leq \abs{w}
  \end{equation}
  This establishes the key inequality
  \begin{equation}\label{eq:main_inequality}
    (q_{k+1} + 2t - 1)q_k - q_{k+1} \leq (\abexp{m} + 2)m - 2.
  \end{equation}
  In other words, our aim is to establish the following proposition.
  
  \begin{proposition}\label{prp:main_inequality}
    Consider a factor $w$ of slope $\alpha$ with minimum abelian period $m$. Let $k$ be the largest integer such that
    $q_k \leq m$, and let $t$ to be the largest integer such that $tq_k \leq m$ with $1 \leq t \leq a_{k+1}$. If
    $m \notin \mathcal{M}_\alpha$, then both \eqref{eq:main_lower_bound} and \eqref{eq:main_inequality} hold.
  \end{proposition}

  As mentioned above, the main point of this paper is to show that left side of \eqref{eq:main_inequality} is so large
  that it also forces $\abexp{m}$ to be relatively large, that is, it forces $\norm{m\alpha}$ to be small, so small
  that $m$ has to correspond to a good rational approximation of $\alpha$.

  \begin{example}
    Recall that the slope $\alpha$ of the Fibonacci word equals $1/\varphi^2$, where $\varphi$ is the Golden ratio.
    Now $\alpha = [0; 2, \overline{1}] \approx 0.38$. The inequality \eqref{eq:main_inequality} predicts that a factor
    of the Fibonacci word having minimum abelian period $9$ must have length at least
    $(13 + 2 \times 1 - 1) \times 8 - 13 = 99$. On the other hand, we have $\norm{9\alpha} \approx 0.44$, so
    $\abexp{9} = \floor{1/\norm{9\alpha}} = 2$ by \autoref{prp:exp_formula}. Thus the upper bound of
    \eqref{eq:main_inequality} is $(2 + 2) \times 9 - 2 = 34$. The conclusion is that there is no factor with minimum
    abelian period $9$ in the Fibonacci word.
  \end{example}
  
  \begin{claim}\label{cl:s_occurrences}
    The word $w$ contains at least $q_k$ occurrences of the singular factor $s$ of length $q_k$.
  \end{claim}
  \begin{proof}
    By \autoref{rem:singular}, all factors of length $q_k$ belong to the same abelian equivalence class except the
    singular factor $s$ of length $q_k$. Thus whenever we factorize a factor of slope $\alpha$ of length $nq_k$ as a
    product $u_1 \dotsm u_n$ with $\abs{u_1} = \ldots = \abs{u_n} = q_k$ and none of the words $u_i$ equal $s$, then
    $u_1 \dotsm u_n$ is an abelian power of period $q_k$ and exponent $n$. The word $w$ cannot be a factor of such an
    abelian power $u_1 \dotsm u_n$. This means that $w$ contains the singular factor $s$ of length $q_k$ in all phases
    modulo $q_k$. Otherwise there is a phase which does not contain $s$ or it contains $s$ only partially (a suffix of
    $s$ as a prefix or a prefix of $s$ as a suffix). As the Parikh vectors of the proper prefixes and suffixes of $s$
    are contained in the Parikh vectors of any factor of length $q_k$ by \autoref{lem:singular_properties} (iii), it
    follows that it is possible to cover $w$ with an abelian repetition of period $q_k$. This is contrary to our
    assumptions. Consequently, the word $w$ contains at least $q_k$ occurrences of $s$.
  \end{proof}
   
  The next result is crucial in obtaining a lower bound for $\abs{w}$.

  \begin{lemma}\label{lem:return_times_singular}
    The return times of the singular factor of length $q_k$ are $q_{k+1}$ and $q_{k+2,1}$.
  \end{lemma}
  \begin{proof}
    The interval $[s]$ of the singular factor $s$ is $I(0, -q_k\alpha)$ or $I(-q_k\alpha, 1)$ by
    \autoref{rem:singular}. Let $x \in [s]$. Then the word $\infw{s}_{x,\alpha}$ begins with $s$. The word $s$ occurs
    in $\infw{s}_{x,\alpha}$ at position $n$, $n > 0$, if $\{x + n\alpha\} \in [s]$. The return time of the prefix $s$
    in $\infw{s}_{x,\alpha}$ is determined by the least such $n$. The length of the interval $[s]$ is
    $\norm{q_k\alpha}$, so it must be that the distance between $x$ and $\{x + n\alpha\}$ is less than
    $\norm{q_k\alpha}$, that is, $\norm{n\alpha} < \norm{q_k\alpha}$. By the best approximation property, we thus
    conclude that $n \geq q_{k+1}$. Let $y$ be a point such that $y \in [s]$ and $\norm{y} = \norm{q_{k+1}\alpha}$. If
    $x \in I(-q_k\alpha, y) \subseteq [s]$, then $\{x + q_{k+1}\alpha\} \in [s]$ and $n = q_{k+1}$. Suppose then that
    $\norm{x} < \norm{y}$. Now $\{x + q_{k+1}\alpha\} \notin [s]$, so $n > q_{k+1}$. On the other hand,
    $\{x + q_{k+2,1}\alpha\} \in [s]$ because $\norm{q_{k+2,1}\alpha} = \norm{q_k\alpha} - \norm{q_{k+1}\alpha}$ and
    the distance between $x$ and $\{-q_k\alpha\}$ is at least $\norm{q_k\alpha} - \norm{q_{k+1}\alpha}$. Therefore
    $n \leq q_{k+2,1}$. If $n < q_{k+2,1}$, then both $\{x + n\alpha\}$ and $\{x + q_{k+2,1}\alpha\}$ lie on $[s]$.
    Then we have $q_{k+2,1} - n \geq q_{k+1}$ by the best approximation property. Therefore
    $q_{k+2,1} \geq q_{k+1} + n > 2q_{k+1} > q_{k+1} + q_k = q_{k+2,1}$; a contradiction. The conclusion is that
    $n = q_{k+2,1}$. We are left with the case $x = y$. Recall that the interval $[s]$ is half-open. If $0 \in [s]$,
    then $\{x + q_{k+1}\alpha\} = 0 \in [s]$, and the claim is clear. Otherwise $\{-q_k\alpha\} \in [s]$,
    $\{x + q_{k+2,1}\alpha\} = \{-q_k\alpha\} \in [s]$, and $\{x + q_{k+1}\alpha\} \notin [s]$, so the preceding
    arguments show that $n = q_{k+2,1}$.
  \end{proof}

  In fact, every factor of a Sturmian word has exactly two returns (the return times are distinct), and Sturmian words
  can be characterized as the recurrent infinite words whose each factor has exactly two returns
  \cite{2001:a_characterization_of_sturmian_words_by_return_words,2000:return_words_in_sturmian_and_episturmian_words}.
  \autoref{lem:return_times_singular} states that the maximum return time of the singular factor of length $q_k$ is
  $q_{k+2,1}$. This is in fact the longest return time among all factors of length $q_k$. This is important in
  determining the so-called recurrence quotients of Sturmian words
  \cite{1940:symbolic_dynamics_ii,1999:limit_values_of_the_recurrence_quotient_of_sturmian}.

  Before proceeding any further, we need the following two technical lemmas on complete first returns to $s$ in the
  same phase. The first lemma essentially says that a Sturmian word $\infw{s}$ of slope $\alpha$ is ``covered'' by
  abelian powers of period $q_k$. By this we mean that whenever we factorize $\infw{s}$ as blocks of length $q_k$, the
  singular factor $s$ is seen only rarely. Intuitively, it is ``hard'' to avoid having abelian period $q_k$ for factors
  whose length does not significantly differ from $\abexp{q_k}q_k$. See \autoref{fig:covering} for a picture of two
  such factorizations of the Fibonacci word, one for phase $0$ modulo $5$ and another for phase $1$ modulo $5$. Notice
  that $\abexp{q_k} \geq q_{k+1}$; see \autoref{lem:exponent_lower_bound}.

  \begin{figure}
  \centering
  \begin{tikzpicture}
    \tikzstyle{own}=[line width=0.8pt,black]
    \tikzstyle{point}=[line width=0.5pt,black,fill=white]
    \tikzstyle{point2}=[line width=0.8pt,black]

    \node (box) {%
    \begin{minipage}[t!]{\textwidth}
      \centering
      $\infw{f} = 01001010010010100101 \mathbf{00100} 1010010010100101 \mathbf{00100} 10100101001001010010010100\dotsm$
    \end{minipage}
    };

    \tikzmath{\width = 0.88; \height1 = 0.2; \height2 = -0.25; \pos1 = -6.3; \pos2 = \pos1 + 5*\width; \pos3 =
    -6.13; \pos4 = \pos3 + 9*\width;}

    \foreach \i in {0,...,3} {
      \draw[-] ({\pos1+\i*\width},{\height1}) to[bend left] ({\pos1+(\i + 1)*\width},{\height1});
    }
    \foreach \i in {0,...,9} {
      \draw[-] ({\pos2+\i*\width},{\height1}) to[bend left] ({\pos2+(\i + 1)*\width},{\height1});
    }

    \foreach \i in {0,...,7} {
      \draw[-] ({\pos3+\i*\width},{\height2}) to[bend right] ({\pos3+(\i + 1)*\width},{\height2});
    }
    \foreach \i in {0,...,5} {
      \draw[-] ({\pos4+\i*\width},{\height2}) to[bend right] ({\pos4+(\i + 1)*\width},{\height2});
    }
  \end{tikzpicture}
  \caption{The Fibonacci word $\infw{f}$ factorized as a product of blocks of length $5$ in phases $0$ and $1$ modulo $5$. The singular factor $00100$ of length $5$ is seen only rarely in each phase.}
  \label{fig:covering}
  \end{figure}

  \begin{lemma}\label{lem:covering}
    Let $s$ be the singular factor of length $q_k$ for some $k \geq 0$, and let $w$ in $\Lang{\alpha}$ be a complete
    first return to $s$ in the same phase. Then the word $s^{-1}ws^{-1}$ is an abelian power of period $q_k$ having
    exponent $\abexp{q_k} - 1$ or $\abexp{q_k}$.
  \end{lemma}
  \begin{proof}
    Suppose that $w$ is a prefix of a Sturmian word $\infw{s}_{x,\alpha}$. As $w$ begins with $s$, we have $x \in [s]$.
    By \autoref{rem:singular}, the interval $[s]$ has endpoints $0$ and $\{-q_k\alpha\}$. The occurrences of $s$ in
    $\infw{s}_{x,\alpha}$ in the same phase as the prefix $s$ correspond to points of the form $\{x + nq_k \alpha\}$
    that are interior points of the interval $[s]$.\footnote{Since $\infw{s}_{x,\alpha}$ is recurrent and $\alpha$ is
    irrational, we may assume that none of these points coincide with the endpoints of $[s]$.} Consider the smallest
    such positive $n$ (such a number exists because the sequence $(nq_k\alpha)_n$ is dense in $\T$ by the well-known
    Kronecker Approximation Theorem; see \cite[Ch.~XXIII]{1979:an_introduction_to_the_theory_of_numbers}). The points
    $\{x + q_k\alpha\}$, $\ldots$, $\{x + (n - 1)q_k\alpha\}$ lie on the interval $I(-q_k\alpha, 1)$ (if
    $\{-q_k\alpha\} < 1/2)$ or on the interval $I(0, -q_k\alpha)$ (if $\{-q_k\alpha\} > 1/2$). Thus
    \autoref{prp:ab_eq_geometric_2} implies that the word $s^{-1}ws^{-1}$ is an abelian power of period $q_k$ and
    exponent $n - 1$. Moreover, the exponent is found by adding one to the times the length $\norm{q_k\alpha}$ fits
    into the interval $I(x + q_k\alpha, 1)$ (if $\{-q_k\alpha\} < 1/2)$ or into the interval
    $I(x + q_k\alpha, -q_k\alpha)$ (if $\{-q_k\alpha\} > 1/2$). Suppose that $\{-q_k\alpha\} < 1/2$. If $x$ is
    arbitrarily close to $0$, then $\{x + q_k\alpha\}$ is arbitrarily close to $\{-q\alpha\}$, so in this case
    $n - 1 = 1 + \floor{(1 - \norm{q_k\alpha})/\norm{q_k\alpha}}$, that is, we have $n - 1 = \abexp{q_k}$. If $x$ is
    arbitrarily close to $\{-q_k\alpha\}$, then analogously we have $n - 1 = \abexp{q_k} - 1$. The case
    $\{-q_k\alpha\} > 1/2$ is similar.
  \end{proof}

  \begin{lemma}\label{lem:first_last_letter}
    Let $s$ be the singular factor of length $q_k$ for some $k \geq 0$. Let $w$ in $\Lang{\alpha}$ be a complete first
    return to $s$ in the same phase, and write $s^{-1}ws^{-1} = u_0 u_1 \dotsm u_{n - 1}$ with
    $\abs{u_0} = \ldots = \abs{u_{n - 1}} = q_k$. Then the words $u_0$, $u_1$, $\ldots$, $u_{\lambda-1}$ end with the
    same letter as $s$ and the words $u_{n - \lambda}$, $u_{n - \lambda + 1}$, $\ldots$, $u_{n - 1}$ begin with the
    same letter as $s$ when
    \begin{equation*}
      \lambda = \begin{cases}
                  q_{k+1} - p_{k+1} - 1, &\text{if $k$ is odd}, \\
                  p_{k+1} - 1,           &\text{if $k$ is even}.
                \end{cases}
    \end{equation*}
    Moreover, the singular factor $s$ ends and begins with the same letter.
  \end{lemma}
  \begin{proof}
    If $k = 0$, then the claim is true, so suppose that $k \geq 1$. Let $w$ be a prefix of a Sturmian word
    $\infw{s}_{x,\alpha}$. Then $x \in [s]$, and we can assume without loss of generality that $x$ is an interior point
    of $[s]$. We consider first the latter claim concerning the first letters of the words $u_j$. Assume that
    $\{-q_k\alpha\} > 1/2$ so that $s$ begins with the letter $1$. If the points $\{x + iq_k\alpha\}$,
    $\{x + (i+1)q_k\alpha\}$, $\ldots$, $\{x + (n - 1)q_k\alpha\}$ lie on the interval $[1]$ of length $\alpha$, then
    the words $u_i$, $\ldots$, $u_{n-1}$ begin with the letter $1$. Notice that the distance between two consecutive
    points is $\norm{q_k\alpha}$ and that $\{x + nq_k\alpha\}$ lies on $[1]$. The worst case scenario is that the point
    $\{x + nq_k\alpha\}$ is very close to the point $\{-q_k\alpha\}$, and then it must be that the $n - i$ consecutive
    distances $\norm{q_k\alpha}$ must fit into $\alpha -  \norm{q_k\alpha}$. Suppose next that $\{-q_k\alpha\} < 1/2$.
    In this case, the word $s$ begins with the letter $0$. Similarly we need to see if the points $\{x + iq_k\alpha\}$,
    $\ldots$, $\{x + (n - 1)q_k\alpha\}$ are placed on the interval $[0]$ of length $1 - \alpha$. This time we need to
    check how many times $\norm{q_k\alpha}$ fits into $1 - \alpha - \norm{q_k\alpha}$. Thus $i$ is maximal when
    $n - 1 - i + 1$ equals $\floor{\alpha/\norm{q_k\alpha}} - 1$ (if $\{-q_k\alpha\} > 1/2$) or
    $\floor{(1-\alpha)/\norm{q_k\alpha}} - 1$ (if $\{-q_k\alpha\} < 1/2$). Consider the former case
    $\{-q_k\alpha\} > 1/2$. By applying \eqref{eq:cf_identity}, we obtain that
    \begin{equation*}
      \frac{\alpha}{\norm{q_k\alpha}} = \frac{p_{k+1}\norm{q_k\alpha} + \norm{q_{k+1}\alpha}}{\norm{q_k\alpha}} \geq p_{k+1}.
    \end{equation*}
    Suppose then that $\{-q_k\alpha\} < 1/2$. In this case, we derive using \eqref{eq:cf_identity},
    \eqref{eq:cf_identity_2}, and \eqref{eq:norm_alpha} that
    \begin{align*}
      \frac{1 - \alpha}{\norm{q_k\alpha}} &= \frac{1}{\norm{q_k\alpha}} - p_{k+1} - \frac{\norm{q_{k+1}\alpha}}{\norm{q_k\alpha}} \\
                                          &= \alpha_{k+1}q_k + q_{k-1} - p_{k+1} - \frac{\norm{q_{k+1}\alpha}}{\norm{q_k\alpha}} \\
                                          &= q_{k+1} - p_{k+1} + [0; a_{k+2}, \ldots]q_k - \frac{\norm{q_{k+1}\alpha}}{\norm{q_k\alpha}} \\
                                          &= q_{k+1} - p_{k+1} + \frac{1}{\alpha_{k+2}}q_k - \frac{1}{\alpha_{k+2}} \\
                                          &= q_{k+1} - p_{k+1} + \frac{q_k - 1}{\alpha_{k+2}} \\
                                          &\geq q_{k+1} - p_{k+1}.
    \end{align*}
    Together the two preceding inequalities establish the latter claim on first letters of the words $u_j$.

    Consider then the former claim about the last letters of the words $u_j$. Suppose first that
    $\{-q_k\alpha\} > 1/2$. The final letter of $s$ is determined by the point $\{x + (q_k - 1)\alpha\}$. Since
    $[s] = I(0, -q_k\alpha)$, we have $\{x + q_k\alpha\} \in I(0, q_k\alpha)$, and hence
    $1 - \alpha < \{x + (q_k - 1)\alpha\} < 1$ because $0 < \{q_k\alpha\} < \alpha$. This means that $s$ ends with the
    letter $1$. As long as the points $\{x + 2q_k\alpha\}$, $\ldots$, $\{x + (i + 1)q_k\alpha\}$ lie between the points
    $0$ and $\alpha$, the words $u_0$, $\ldots$, $u_{i-1}$ end with the letter $1$. Again, it is clearly sufficient to
    compute $\floor{\alpha/\norm{q_k\alpha}} - 1$. The final case $\{-q_k\alpha\} < 1/2$ is analogous.
  \end{proof}

  \begin{remark}\label{rem:weak_consequence}
    In the proof of \autoref{lem:first_last_letter}, we derived lower bounds for both $\alpha/\norm{q_k\alpha}$ and
    $(1 - \alpha)/\norm{q_k\alpha}$, the lower bound for the former being $p_{k+1}$. Since $1 - \alpha > \alpha$, we
    derive a common lower bound $p_{k+1}$ for both quantities, that is, $\lambda \geq p_{k+1} - 1$ for all $k \geq 1$.
    It is straightforward to see that $p_k \geq a_k$ for all $k \geq 1$, so we conclude that for all $k \geq 1$ the
    $a_{k+1} - 1$ consecutive factors of length $q_k$ preceding (resp. following) each occurrence of the singular factor
    $s$ of length $q_k$ begin (resp. end) with the same letter as $s$. This is the consequence of
    \autoref{lem:first_last_letter} we need in this paper.
  \end{remark}

  We may now continue to derive the inequality \eqref{eq:main_inequality}. The factor $w$ contains at least $q_k$
  occurrences of $s$ (\autoref{cl:s_occurrences}) and, by \autoref{lem:return_times_singular}, the minimum return time
  of $s$ is $q_{k+1}$. Thus if $w$ contains at least $q_k + 2$ occurrences of $s$ then,
  $\abs{w} \geq (q_k + 1)q_{k+1} + q_k$. In this case \eqref{eq:main_lower_bound} holds as it is straightforward to
  compute that $(q_k + 1)q_{k+1} + q_k - ((q_{k+1} + 2t - 1)q_k - q_{k+1}) > 0$ (recall from the paragraph following
  \autoref{thm:main} that $t \leq a_{k+1}$). We may thus assume that $w$ contains
  at most $q_k + 1$ occurrences of $s$. Suppose then that $t = 1$. If $w$ contains $q_k$ occurrences of $s$, then we
  have $\abs{w} \geq (q_k - 1)q_{k+1} + q_k$, that is, \eqref{eq:main_lower_bound} holds. If $w$ contains $q_k + 1$
  occurrences of $s$ then, by the same logic, we see that $\abs{w} \geq q_k q_{k+1} + q_k$. Then
  \eqref{eq:main_lower_bound} is true as is easy to verify. Therefore we may assume that $t > 1$.

  Suppose first that $w$ contains exactly $q_k + 1$ occurrences of $s$. By \autoref{lem:return_times_singular}, the
  occurrences of $s$ do not overlap, so we may write $w = u_0 s u_1 s u_2 \dotsm s u_{q_k + 1}$ for some words $u_0$,
  $u_1$, $\ldots$, $u_{q_k + 1}$. Since both return times of $s$, $q_{k+1}$ and $q_{k+2,1}$, equal $q_{k-1}$ modulo
  $q_k$, it follows that the occurrences of $s$ in $w$ are all in different phases modulo $q_k$ except the first and
  final one that are in the same phase (it is easy to see by induction that $\gcd(q_{k-1}, q_k) = 1$). By
  \autoref{lem:covering}, the word $(u_0 s)^{-1} w (s u_{q_k + 1})^{-1}$ is an abelian power of period $q_k$ and
  exponent $E$, where $E$ equals $\abexp{q_k} - 1$ or $\abexp{q_k}$. Write
  $(u_0 s)^{-1} w (s u_{q_k + 1})^{-1} = \beta_0 \dotsm \beta_{E-1}$ with
  $\abs{\beta_0} = \ldots = \abs{\beta_{E-1}} = q_k$. Consider the words $\beta_0 \dotsm \beta_{E-1}$,
  $\beta_1 \dotsm \beta_{E-1}$, $\ldots$, $\beta_{t-1} \dotsm \beta_{E-1}$. These words can be viewed as (possibly
  degenerate) abelian repetitions with period $tq_k$ and empty head.

  \begin{claim}\label{cl:extension}
    The abelian repetitions $\beta_0 \dotsm \beta_{E-1}$, $\beta_1 \dotsm \beta_{E-1}$, $\ldots$,
    $\beta_{t-1} \dotsm \beta_{E-1}$ of period $tq_k$ can be extended to have head and tail of maximum length
    $tq_k - 1$.
  \end{claim}
  \begin{proof}
    Consider the word $\beta_n \dotsm \beta_{E-1}$ with $0 \leq n \leq t - 1$. Suppose that the singular factor $s$
    begins with letter $x$, and let $y$ to be the letter such that $y \neq x$. The word $s$ also ends with
    the letter $x$ by \autoref{lem:first_last_letter}. Let $i$ equal the number of letters $x$ in the nonsingular
    factors of length $q_k$ and set $j = q_k - i$. The word $\beta_0 \dotsm \beta_{E-1}$ is preceded by the word $s$,
    and $\abs{s}_x = i + 1$ and $\abs{s}_y = j - 1$. Further, the word $s \beta_0 \dotsm \beta_{E-1}$ is preceded by an
    abelian power of period $q_k$ and exponent at least $\abexp{q_k} - 1$. This power might extend beyond the starting
    position of $w$. For our purposes, it is irrelevant how $w$ is extended to the left; all that matters is that by
    recurrence the left extension of $w$ exists in $\Lang{\alpha}$. We conclude that the first $s$ of $w$ is preceded
    by an abelian power $\gamma_0 \dotsm \gamma_{t - (n + 2)}$ of period $q_k$ with
    $\abs{\gamma_0} = \ldots = \abs{\gamma_{t - (n + 2)}} = q_k$. By \autoref{rem:weak_consequence}, the words
    $\gamma_0$, $\ldots$, $\gamma_{t - (n + 2)}$ all begin with the letter $x$. Thus
    \begin{equation*}
      \abs{x^{-1}\gamma_0 \dotsm \gamma_{t - (n + 2)} s \beta_0 \dotsm \beta_{n-1}}_x = ((t - (n + 1))i - 1) + (i + 1) + ni = ti
    \end{equation*}
    and
    \begin{equation*}
      \abs{x^{-1}\gamma_0 \dotsm \gamma_{t - (n + 2)} s \beta_0 \dotsm \beta_{n-1}}_y = tj - 1.
    \end{equation*}
    Therefore the Parikh vector of the factor $x^{-1}\gamma_0 \dotsm \gamma_{t - (n + 2)} s \beta_0 \dotsm \beta_{n-1}$
    is contained in the Parikh vector of $\beta_n \dotsm \beta_{n + t - 1}$. Thus the abelian repetition 
    $\beta_n \dotsm \beta_{E-1}$ can be extended to have a head of maximal length $tq_k - 1$.

    Let $r$ be the largest integer such that $rt \leq E - n$. Then the abelian repetition $\beta_n \dotsm \beta_{E-1}$
    has tail $\beta_{n + rt} \dotsm \beta_{E-1}$. This tail is followed by the word $s$, which is in turn followed by
    an abelian power $\delta_0 \dotsm \delta_{(r+1)t - E + n - 2}$ of period $q_k$ with
    $\abs{\delta_0} = \ldots = \abs{\delta_{(r+1)t - E + n - 2}} = q_k$. By \autoref{rem:weak_consequence}, the words
    $\delta_0$, $\ldots$, $\delta_{(r+1)t - E + n - 2}$ end with the letter $x$. Similarly to above, if we remove the
    final letter of the word $\beta_{n + rt} \dotsm \beta_{E-1} s \delta_0 \dotsm \delta_{(r+1)t - E + n - 2}$ it will
    cancel the additional letter $x$ in $s$, and we see that the tail of $\beta_n \dotsm \beta_{E-1}$ can be extended
    to maximum length $tq_k - 1$.
  \end{proof}

  Let then $n$ be an integer such that $0 \leq n \leq t - 1$, and let $\lambda_n$ to be the abelian repetition
  $\beta_n \dotsm \beta_{E-1}$ extended to have head and tail of length $tq_k - 1$. We define the \emph{left overhang}
  of $\lambda_n$, denoted by $L(\lambda_n)$, to be its prefix that comes before the occurrence of $s$ that coincides
  with the first $s$ in $w$ if it exists; otherwise we set $L(\lambda_n) = \varepsilon$. Using the notation of the
  proof of \autoref{cl:extension}, we thus set $L(\lambda_n) = x^{-1} \gamma_0 \dotsm \gamma_{t - (n + 2)}$ when
  $n < t - 1$. Similarly, we define the word $R(\lambda_n)$, the \emph{right overhang} of $\lambda_n$, as the suffix of
  $\lambda_n$ that comes after the final $s$ in $w$ if it exists, that is, $R(\lambda_n) = \delta_0 \dotsm
  \delta_{(r+1)t - E + n - 2}x^{-1}$, $n \neq E - (r + 1)t + 2$, in the notation of the proof of
  \autoref{cl:extension}. Here $r$ is the largest integer such that $rt \leq E - n$. In particular, we have
  \begin{equation*}
    \abs{L(\lambda_n)} = (t - (n + 1))q_k - 1
  \end{equation*}
  and
  \begin{equation*}
    \abs{R(\lambda_n)} = ((r + 1)t - E + n - 1)q_k - 1
  \end{equation*}
  when $L(\lambda_n)$ and $R(\lambda_n)$ are nonempty.

  Since $tq_k < m$ and $m$ is the minimum abelian period of $w$, none of the words $\lambda_0$, $\ldots$,
  $\lambda_{t-1}$ can completely cover $w$. Let $i$ be the largest integer such that $iq_k \leq \abs{u_0}$. If
  $i \geq t - 1$, then $\abs{w} \geq (t - 1)q_k + q_k q_{k+1} + q_k$, and it is elementary to verify that
  \eqref{eq:main_lower_bound} holds. We hence assume that $i < t - 1$. Now the left overhangs of the words $\lambda_0$,
  $\ldots$, $\lambda_{t - (i + 1) - 1}$ cover $u_0$, so it must be that $w$ extends beyond their right overhangs. At
  least one of these right overhangs must have length at least $(t - i - 2)q_k - 1$. Namely if $\abs{R(\lambda_0)}$ is
  as small as possible, then $R(\lambda_0) = \varepsilon$ and
  $\abs{R(\lambda_{t - (i + 1) - 1})} = (t - i - 2)q_k - 1$. It follows that
  \begin{align*}
    \abs{w} &\geq \abs{u_0} + q_k q_{k+1} + q_k + (t - i - 2)q_k \\
            &\geq iq_k + (q_{k+1} + 1)q_k + (t - i - 2)q_k \\
            &= (q_{k+1} + t - 1)q_k,
  \end{align*}
  and it is again straightforward to verify that \eqref{eq:main_lower_bound} holds.

  Suppose finally that $w$ contains exactly $q_k$ occurrences of $s$. Factor $w$ again according to the occurrences of
  $s$: $w = u_0 s u_1 \dotsm s u_{q_k}$. Similarly as before, the factor $(u_0 s)^{-1} w$ is now a prefix of an abelian
  power $\beta_0 \dotsm \beta_{E-1}$. The arguments of \autoref{cl:extension} can now be repeated to see that the
  abelian power $\beta_0 \dotsm \beta_{E-1}$ can be extended to an abelian repetition with period $tq_k$ and head of
  length $tq_k - 1$. Since $m$ is the minimum abelian period of $w$, this head cannot cover $u_0$ completely, so
  $\abs{u_0} \geq (t - 1)q_k$. Consider next the reversal $\mirror{w}$ of $w$. Since the word $s$ is a palindrome by
  \autoref{lem:singular_properties} (ii), we have $\mirror{w} = \mirror{u}_{q_k} s \dotsm \mirror{u}_1 s \mirror{u}_0$.
  Since the language $\Lang{\alpha}$ is closed under reversal, we have $\mirror{w} \in \Lang{\alpha}$. The minimum
  abelian period is invariant under reversal so, by repeating the preceding arguments, we see that
  $\abs{u_{q_k}} \geq (t - 1)q_k$. A short computation shows that \eqref{eq:main_lower_bound} holds also in this final
  case. This ends the proof of \autoref{prp:main_inequality}.

  \section{Proofs of the Main Results}\label{sec:main_proofs}
  In this section, we prove the main results, Theorems \ref{thm:main} and \ref{thm:fibonacci_characterization}.
  Throughout this section, we continue to use the notation of \autoref{sec:main_inequality}. We consider an abelian
  period $m$ of a factor $w$ of slope $\alpha$, and we assume that $m \notin \qkl{\alpha} \cup \mathcal{M}_\alpha$. We
  let $k$ to be the largest integer such that $q_k < m$ and $t$ to be the largest integer such that $tq_k < m$ with
  $1 \leq t \leq a_{k+1}$. Recall that $k \geq 1$. The assumption $m \notin \qkl{\alpha} \cup \mathcal{M}_\alpha$
  implies that $\norm{m\alpha} > \norm{q_{k-1}\alpha} + \norm{q_k\alpha}$. This is most easily seen from
  \autoref{fig:intervals}. Consider first the lower part of the figure. Point with distance at most 
  $\norm{q_{k-1}\alpha} + \norm{q_k\alpha}$ is either in $\mathcal{M}_\alpha$ or equals $\{-q_{k,a_k-1}\alpha\}$. The
  latter option is ruled out because $m > q_k > q_{k,a_k-1}$. In the upper part, all points with distance at most
  $\norm{q_{k-1}\alpha}$ are in $\qkl{\alpha}$. Since $q_k < m < q_{k+1}$, the best approximation property shows that
  the distance from $\{-m\alpha\}$ to $\{-q_{k-1}\alpha\}$ is greater than $\norm{q_k\alpha}$. The claim follows.

  The next lemma is used repeatedly in the following proofs. It will be used to show \autoref{thm:main} to be true in
  the fairly typical case $\abexp{m} < q_k - 1$; the remaining case $\abexp{m} = q_k - 1$ is the difficult one.

  \begin{lemma}\label{lem:improved_ub}
    Consider a factor $w$ of slope $\alpha$ with abelian period $m$ and exponent $E$. Let $k$ be the largest integer
    such that $q_k \leq m$. If $m \notin \mathcal{M}_\alpha$ and $E < q_k - 1$, then $m$ is not the minimum abelian
    period of $w$.
  \end{lemma}
  \begin{proof}
    If $k = 0$, then there is nothing to prove as $E$ is always positive. Assume that $k \geq 1$,
    $m \notin \mathcal{M}_\alpha$, and $E < q_k - 1$. Clearly $\abs{w} \leq (E + 2)m - 2$. Suppose for a contradiction
    that $m$ is the minimum abelian period of $w$. Then \autoref{prp:main_inequality} gives
    \begin{equation}\label{eq:improved_inequality}
      (q_{k+1} + 2t - 1)q_k - q_{k+1} \leq (E + 2)m - 2.
    \end{equation}
    Let us assume first that $t < a_{k+1}$. Using the upper bound $m < (t + 1)q_k$, we obtain from
    \eqref{eq:improved_inequality} that
    \begin{equation*}
      (q_{k+1} + 2t - 1)q_k - q_{k+1} \leq q_k((t + 1)q_k - 1) - 2.
    \end{equation*}
    Using the equality $q_{k+1} = a_{k+1}q_k + q_{k-1}$, we obtain by rearrangement the equivalent inequality
    \begin{equation}\label{eq:foo2}
      ((a_{k+1} - (t + 1))q_k + q_{k-1} + 2t - a_{k+1})q_k \leq q_{k-1} - 2.
    \end{equation}
    The right side of \eqref{eq:foo2} is at less than $q_k$, so it must be that the coefficient of $q_k$ on the left is
    at most $0$. If $k \geq 2$, then $q_k, q_{k-1} \geq q_1 \geq 2$, and we obtain
    \begin{align*}
      (a_{k+1} - (t + 1))q_k + q_{k-1} + 2t - a_{k+1} &\geq 2(a_{k+1} - (t + 1)) + q_{k-1} + 2t - a_{k+1} \\
                                                      &= a_{k+1} + q_{k-1} - 2 \\
                                                      &\geq a_{k+1},
    \end{align*}
    which is impossible as $a_{k+1} \geq 1$. If $k = 1$, then the right side of \eqref{eq:foo2} is negative. Analogous
    calculation now gives $(a_2 - (t + 1))q_1 + q_0 + 2t - a_2 \geq a_2 + q_0 - 2 = a_2 - 1 \geq 0$, which is again
    contradictory.
    
    Suppose then that $t = a_{k+1}$. Now the upper bound $m < q_{k+1}$ and
    \eqref{eq:improved_inequality} give
    \begin{equation*}
      (q_{k+1} + 2a_{k+1} - 1)q_k - q_{k+1} \leq q_k(q_{k+1} - 1) - 2.
    \end{equation*}
    Rearranging like above then gives
    \begin{equation*}
      2a_{k+1}q_k \leq q_{k+1} = a_{k+1}q_k + q_{k-1},
    \end{equation*}
    which is again impossible since $q_k > q_{k-1}$.
  \end{proof}

  Next we do not make restrictions on $\abexp{m}$ and prove \autoref{thm:main} in almost every case.

  \begin{lemma}\label{lem:almost_all}
    If $k \geq 4$, then $m$ is not the minimum abelian period of $w$.
  \end{lemma}
  \begin{proof}
    Suppose for a contradiction that $m$ is the minimum abelian period of $w$. Since
    $\norm{m\alpha} > \norm{q_{k-1}\alpha} + \norm{q_k\alpha}$, we have $\abexp{m} < q_k$ by
    \autoref{lem:relation_exponent_convergent}. Thus we obtain from \autoref{prp:main_inequality} that
    \begin{equation*}
      (q_{k+1} + 2t - 1)q_k - q_{k+1} \leq (q_k + 1)m - 2.
    \end{equation*}
    Consider first the case $t < a_{k+1}$. We obtain from the previous inequality that
    \begin{equation*}
      (q_{k+1} + 2t - 1)q_k - q_{k+1} \leq (q_k + 1)((t + 1)q_k - 1) - 2.
    \end{equation*}
    By rearranging and writing $q_{k+1} = a_{k+1}q_k + q_{k-1}$, we obtain
    \begin{equation}\label{eq:foo1}
      ((a_{k+1} - (t + 1))q_k + q_{k-1} + t - 1)q_k \leq q_{k+1} - 3.
    \end{equation}
    In order for this inequality to hold, it is necessary that the coefficient of $q_k$ on the left side is at most
    $a_{k+1}$. Since $k > 1$, we have $q_k \geq q_2 \geq 3$ and $q_{k-1} \geq q_1 \geq 2$. It follows that
    \begin{equation*}
      a_{k+1} \geq 3(a_{k+1} - t - 1) + t + 1,
    \end{equation*}
    which yields $a_{k+1} \leq t + 1$. This is true only if $t = a_{k+1} - 1$. It is also clear from \eqref{eq:foo1}
    that we must have $q_{k-1} \leq 2$ if $t = a_{k+1} - 1$. This means that $k = 2$ and $a_1 = 2$. However, now the
    left side of \eqref{eq:foo1} is $a_3 q_2$ and the right side is $a_3 q_2 + q_1 - 3$. This is impossible as
    $q_1 = 2$.

    Suppose then that $t = a_{k+1}$. Let us first assume that $m \leq (a_{k+1}-1)q_k + 2q_{k-1}$ to obtain from
    \eqref{eq:main_inequality} that
    \begin{equation*}
      (q_{k+1} + 2a_{k+1} - 1)q_k - q_{k+1} \leq (q_k + 1)((a_{k+1}-1)q_k + 2q_{k-1}) - 2.
    \end{equation*}
    This inequality is equivalent to
    \begin{equation*}
      (q_k - q_{k-1} + a_{k+1})q_k \leq q_{k+1} + 2q_{k-1} - 2.
    \end{equation*}
    Substituting $q_{k+1} = a_{k+1}q_k + q_{k-1}$ gives the equivalent inequality
    \begin{equation}\label{eq:offending}
      (q_k - q_{k-1})q_k \leq 3q_{k-1} - 2.
    \end{equation}
    Since $k \geq 4$, we have $q_k - q_{k-1} \geq q_4 - q_3 \geq q_2 \geq 3$. This together with \eqref{eq:offending}
    gives $3q_k < 3q_{k-1}$, which is obviously false. We may thus assume that $m > (a_{k+1}-1)q_k + 2q_{k-1}$. Then
    \autoref{lem:larger_lower_bound} implies that
    $\norm{m\alpha} \geq \norm{q_{k-1}\alpha} + (a_{k+1} + 1)\norm{q_k\alpha}$. This means that we can improve the bound
    $\abexp{m} < q_k$ to $\abexp{m} < q_k - 1$ by \autoref{lem:relation_exponent_convergent_2}. The desired
    contradiction follows now from \autoref{lem:improved_ub}.
  \end{proof}

  By \autoref{lem:almost_all}, we are left with the cases $k = 1$, $k = 2$, and $k = 3$. However in the last two cases,
  the arguments of the proof of \autoref{lem:almost_all} apply under suitable conditions. Let us analyze the situation.
  The only place where we really need the assumption $k \geq 4$ is when a contradiction is derived from
  \eqref{eq:offending}. Here we needed $k \geq 4$ to establish that $q_k - q_{k-1} \geq 3$. Let us see when
  $q_k - q_{k-1} \leq 2$. Now $q_k - q_{k-1} = (a_k - 1)q_{k-1} + q_{k-2}$, so if $q_k - q_{k-1} \leq 2$, then it must
  be that $a_k = 1$ whenever $q_{k-1} \geq 2$. Moreover, by \autoref{lem:relation_exponent_convergent_2}, we do not
  need to know that $q_k - q_{k-1} \geq 3$ in order to improve the bound to $\abexp{m} < q_k - 1$ if we know that
  $m \neq (a_{k+1} - 1)q_k + 2q_{k-1}$. We are thus left to prove \autoref{thm:main} in the following cases:
  \begin{itemize}
    \item $k = 1$;
    \item $m = (a_{k+1} - 1)q_k + 2q_{k-1}$, $a_k = 1$ when $k = 2$ or $k = 3$.
  \end{itemize}

  It must be emphasized that the proof of \autoref{lem:almost_all} does not work in these final cases: the inequality
  \eqref{eq:main_inequality} is not enough. Indeed, if $\alpha = [0; 2, \overline{1}]$ and
  $m = (a_3 - 1)q_2 + 2q_1 = 4$, then $\abexp{m} = 2$ and the left side of \eqref{eq:main_inequality} is $13$, but the
  right side is $14$. A more interesting example is perhaps $\alpha = [0; 2, 3, 1, 6, \overline{1}]$. When $k = 3$, we
  have $m = (a_4 - 1)q_3 + 2q_2 = (6 - 1) \times 9 + 2 \times 7 = 59$ and $\abexp{m} = 8 = q_3 - 1$. The left side of
  \eqref{eq:main_inequality} is $(61 + 2 \times 6 - 1) \times 9 - 61 = 587$, and the right side is $588$.

  The following general lemma handles the above cases $k = 2, 3$ and a subcase of $k = 1$.

  \begin{lemma}\label{lem:cases_2_3}
    If either
    \begin{enumerate}[(i)]
      \item $k = 1$, $a_1 \geq 3$, $a_2 > 1$, and $m = (a_2 - 1)q_1 + 2q_0$ or
      \item $k \geq 2$, $a_k = 1$, and $m = (a_{k+1} - 1)q_k + 2q_{k-1}$,
    \end{enumerate}
    then $m$ is not the minimum abelian period of $w$.
  \end{lemma}
  \begin{proof}
    Let $m = (a_{k+1} - 1)q_k + 2q_{k-1}$ for $k \geq 1$, and assume for a contradiction that there exists a factor $w$
    of slope $\alpha$ with minimum abelian period $m$. Notice that we have $(a_{k+1} - 1)q_k < m < a_{k+1}q_k$ in the
    case (i), and $a_{k+1}q_k < m < q_{k+1}$ in the case (ii). Write $w = u_0 s u_1 \dotsm s u_\lambda$ according to
    the $\lambda$ occurrences of the singular factor $s$ of length $q_k$ in $w$. Recall that $\lambda \geq q_k$ by
    \autoref{cl:s_occurrences}. Factorize $w = \beta_0 \beta_1 \dotsm \beta_E \beta_{E+1}$, with
    $\abs{\beta_1} = \ldots = \abs{\beta_E} = m$, according to the minimum abelian period $m$ of $w$. We may assume
    that $E = q_k - 1$ for otherwise the claim is clear by \autoref{lem:improved_ub}. Suppose first that
    $\lambda > q_k$. Then
    \begin{equation*}
      \abs{w} - \abs{u_0} \geq q_k q_{k+1} + q_k
    \end{equation*}
    according to \autoref{lem:return_times_singular}. Assume then that $\lambda = q_k$. When the inequality
    \eqref{eq:main_lower_bound} was derived, we showed that $\abs{u_0}, \abs{u_\lambda} \geq (a_{k+1} - 1)q_k$.
    In particular, we have $\abs{u_\lambda} \geq (a_{k+1} - 1)q_k$, and thus
    \begin{equation}\label{eq:bar_1}
      \abs{w} - \abs{u_0} \geq (q_k - 1)q_{k+1} + a_{k+1}q_k.
    \end{equation}
    Thus no matter the value of $\lambda$, the inequality \eqref{eq:bar_1} holds.

    Next we want to show that $\abs{\beta_0} > \abs{u_0}$. Assume on the contrary that $\abs{\beta_0} \leq \abs{u_0}$.
    Then
    \begin{align*}
      \abs{\beta_{E+1}} &= \abs{w} - \abs{\beta_0 \beta_1 \dotsm \beta_E} \\
                        &\geq \abs{w} - \abs{u_0} - \abs{\beta_1 \dotsm \beta_E} \\
                        &\geq (q_k - 1)q_{k+1} + a_{k+1}q_k - Em \\
                        &= (q_k - 1)q_{k+1} + q_k - (q_k - 1)((a_{k+1} - 1)q_k + 2q_{k-1}) + (a_{k+1} - 1)q_k \\
                        &= (q_{k+1} + 1)q_k - q_{k+1} - (a_{k+1} - 1)q_k q_k - 2q_{k-1}q_k + (a_{k+1} - 1)q_k + 2q_{k-1} + (a_{k+1} - 1)q_k \\
                        &= (q_{k+1} + 1 - (a_{k+1} - 1)q_k - 2q_{k-1})q_k - q_{k+1} + (a_{k+1} - 1)q_k + 2q_{k-1} + (a_{k+1} - 1)q_k \\
                        &= (q_k - q_{k-1} + 1)q_k - q_k + q_{k-1} + (a_{k+1} - 1)q_k.
    \end{align*}
    Let $Q = (q_k - q_{k-1} + 1)q_k - q_k + q_{k-1} + (a_{k+1} - 1)q_k$. If (i) holds, then
    \begin{equation*}
      Q = (q_1 - 2)q_1 + q_0 + (a_2 - 1)q_1 > 2q_0 + (a_2 - 1)q_1 = m.
    \end{equation*}
    If (ii) holds, then
    \begin{align*}
      Q &= (q_{k-2} + 1)q_k - q_{k-2} + (a_{k+1} - 1)q_k \\
        &= (q_{k-2} + 1)(q_{k-1} + q_{k-2}) - q_{k-2} + (a_{k+1} - 1)q_k \\
        &= (q_{k-2} + 1)q_{k-1} + q_{k-2}^2 + (a_{k+1} - 1)q_k \\
        &> 2q_{k-1} + (a_{k+1} - 1)q_k \\
        &= m.
    \end{align*}
    Thus $\abs{\beta_{E+1}} \geq Q > m$. This is a contradiction, so we conclude that $\abs{\beta_0} > \abs{u_0}$.

    We let $J$ denote the longer of the two intervals separated by the points $0$ and $\{-m\alpha\}$. Let
    $x \in [u_0^{-1} w]$. In particular, we have $x \in [s]$. Now $[s] = I(0, -q_k\alpha)$ or $[s] = I(-q_k\alpha, 1)$.
    Let $L = \norm{q_{k-1}\alpha} + \norm{q_k\alpha} + \norm{q_{k+1}\alpha}$. The distance from $\{-m\alpha\}$ to $0$
    through the point $\{q_{k-1}\alpha\}$ equals $L$; see \autoref{fig:intervals}. Our aim is to show that
    $L < \tfrac12$. This establishes that the point $x$ is not on the same side of $0$ as $\{-m\alpha\}$ and that
    $\norm{m\alpha} = L$.
    
    Suppose first that $k \geq 2$. We have $\norm{q_{t+1}\alpha} < \tfrac12 \norm{q_t\alpha}$ for all $t \geq 0$, so
    \begin{equation*}
      L < \frac12(\norm{q_{k-2}\alpha} + \norm{q_{k-1}\alpha} + \norm{q_k\alpha}) \leq \frac12(\norm{q_{k-2}\alpha} + \norm{q_{k-2}\alpha}) = \norm{q_{k-2}\alpha} \leq \alpha < \frac12.
    \end{equation*}
    Assume then that $k = 1$. Since $a_2 > 1$, we have $m \neq 2$. The distance from $\{-m\alpha\}$ to $\{-\alpha\}$ is
    $\norm{q_1\alpha} + \norm{q_2\alpha}$, and the distance from $\{-m\alpha\}$ to $\{-2\alpha\}$ is
    $(a_2 - 1)\norm{q_1\alpha}$. Since $a_1 \geq 3$, we see that
    $1 - L = (a_2 - 1)\norm{q_1\alpha} + (a_1 - 2)\norm{\alpha} + \norm{q_1\alpha} = (a_1 - 2)\norm{\alpha} + a_2 \norm{q_1\alpha}$.
    Clearly $1 - L > L$, so $L < \tfrac12$.

    Since the point $x$ is not on the same side of $0$ as $\{-m\alpha\}$, it follows that $x \in J$. Set
    $y = \{x + (\abs{\beta_0} - \abs{u_0})\alpha\}$. Then $x \neq y$. The abelian power $\beta_1 \dotsm \beta_E$
    beginning at position $\abs{\beta_0}$ of $w$ is not degenerate: $E = q_k - 1 \geq 2$ when $k = 1$ and
    $E = q_k - 1 \geq q_2 - 1 \geq 2$ when $k \geq 2$. Therefore, by \autoref{prp:ab_eq_geometric_2}, the point $y$
    must also lie on $J$. Let $D_1$ be the distance of $y$ to $0$ through the point $\{-m\alpha\}$ and $D_2$ be the
    distance of $y$ to $0$ to the other direction. Since $y$ lies on $J$, it follows that
    $D_1 \geq \norm{m\alpha} \geq \norm{q_{k-1}\alpha}$.
    
    Our next aim is to find a lower bound to the distance $D$ between $x$ and $y$. Notice that since
    $\abs{\beta_0} < m < q_{k+1}$, it must be that $\abs{\beta_0} - \abs{u_0} < q_{k+1}$. It thus follows from the best
    approximation property that $D > \norm{q_{k+1}\alpha}$. If $D = \norm{q_k\alpha}$, then it must be that
    $\abs{\beta_0} - \abs{u_0} = q_k$. This is however not the case because then $y$ would be on the same side of
    $0$ as the point $\{-m\alpha\}$. Therefore $D > \norm{q_k\alpha}$. Since $y \in J$, it follows that $y \notin [s]$.
    In particular, we have $D_2 \geq D$. Since $D > \norm{q_k\alpha}$, it follows from the best approximation
    property that $D \geq \norm{q_{k-1}\alpha}$. The conclusion is that
    $\norm{y} \geq \min\{D_1, D_2\} \geq \norm{q_{k-1}\alpha}$.

    By \autoref{prp:ab_eq_geometric_2}, the exponent $E$ of the abelian power $\beta_1 \dotsm \beta_E$ is the integer
    part of
    \begin{equation*}
      \frac{1 - \norm{m\alpha} - \norm{y}}{\norm{m\alpha}} + 1.
    \end{equation*}
    By the above, we obtain that
    \begin{equation*}
      \frac{1 - \norm{m\alpha} - \norm{y}}{\norm{m\alpha}} + 1 \leq \frac{1 - \norm{q_{k-1}\alpha}}{\norm{m\alpha}}.
    \end{equation*}
    Let $\mathcal{S}$ denote the right side of this inequality. We shall argue that $\mathcal{S} < q_k - 1$. This shows
    that $E < q_k - 1$ contradicting the maximality of $E$ and ending the proof.

    Recall that $\norm{m\alpha} = \norm{q_{k-1}\alpha} + \norm{q_k\alpha} + \norm{q_{k+1}\alpha}$. In particular, we
    have $\norm{m\alpha} > \norm{q_{k-1}\alpha} + \norm{q_k\alpha}$. Hence
    \begin{equation*}
      \mathcal{S} < \frac{1 - \norm{q_{k-1}\alpha}}{\norm{q_{k-1}\alpha} + \norm{q_k\alpha}}.
    \end{equation*}
    By \autoref{lem:consecutive_norm}, we have
    \begin{equation*}
      \frac{1}{\norm{q_{k-1}\alpha} + \norm{q_k\alpha}} = \frac{\alpha_{k+1}q_k + q_{k-1}}{\alpha_{k+1} + 1} \quad \text{and} \quad q_k - \frac{\alpha_{k+1}q_k + q_{k-1}}{\alpha_{k+1} + 1} = \frac{q_k - q_{k-1}}{\alpha_{k+1} + 1},
    \end{equation*}
    so
    \begin{align*}
      &q_k - \frac{1}{\alpha_{k+1} + 1}\Big( q_k - q_{k-1} + \norm{q_{k-1}\alpha}(\alpha_{k+1}q_k + q_{k-1})\Big) \\
      &= q_k - \frac{q_k - q_{k-1}}{\alpha_{k+1} + 1} - \frac{\norm{q_{k-1}\alpha}(\alpha_{k+1}q_k + q_{k-1})}{\alpha_{k+1} + 1} \\
      &= \frac{\alpha_{k+1}q_k + q_{k-1}}{\alpha_{k+1} + 1} - \frac{\norm{q_{k-1}\alpha}(\alpha_{k+1}q_k + q_{k-1})}{\alpha_{k+1} + 1} \\
      &= \frac{1 - \norm{q_{k-1}\alpha}}{\norm{q_{k-1}\alpha} + \norm{q_k\alpha}}.
    \end{align*}
    In other words, we have
    \begin{equation*}
      \mathcal{S} < q_k - \frac{1}{\alpha_{k+1} + 1}\Big( q_k - q_{k-1} + \norm{q_{k-1}\alpha}(\alpha_{k+1}q_k + q_{k-1})\Big).
    \end{equation*}
    Now $\norm{q_{k-1}} = (\alpha_k q_{k-1} + q_{k-2})^{-1}$ by \eqref{eq:norm_alpha}, so
    \begin{equation}\label{eq:bar_2}
      \mathcal{S} < q_k - \frac{1}{\alpha_{k+1} + 1}\left( q_k - q_{k-1} + \frac{\alpha_{k+1}q_k + q_{k-1}}{\alpha_k q_{k-1} + q_{k-2}} \right).
    \end{equation}
    Clearly,
    \begin{equation*}
      \frac{\alpha_{k+1}q_k + q_{k-1}}{\alpha_k q_{k-1} + q_{k-2}} = \frac{\alpha_{k+1}q_k + q_{k-1}}{q_k + \alpha_{k+1}^{-1}q_{k-1}} = \alpha_{k+1},
    \end{equation*}
    so it follows from \eqref{eq:bar_2} that $\mathcal{S} < q_k - 1$ because $q_k - q_{k-1} \geq 1$.
  \end{proof}

  By Lemmas \ref{lem:almost_all} and \ref{lem:cases_2_3}, \autoref{thm:main} is true when $k \geq 2$, and we are left
  with the case $k = 1$. Let us first make some general observations.

  Suppose that $k = 1$ and $m$ is the minimum abelian period of $w$. If $t = a_2$, then
  $a_2 q_1 < m < q_2 = a_2 q_1 + 1$, and such $m$ cannot exist. Thus we may assume that $a_2 > 1$ and $t < a_2$. Now
  $m \neq tq_1 + 1 = q_{2,t}$, so $m$ must equal one of the numbers $tq_1 + 2$, $\ldots$, $tq_1 + q_1 - 1$. In
  particular, it must be that $q_1 = a_1 \geq 3$. The computation at the beginning of the proof of
  \autoref{lem:almost_all} leading to the inequality \eqref{eq:foo1} shows that
  \begin{equation*}
    ((a_2 - (t + 1))q_1 + t)q_1 \leq q_2 - 3.
  \end{equation*}
  The coefficient $(a_2 - (t + 1))q_1 + t$ of $q_1$ on the left is at most $a_2$ and at least $3a_2 - 2t - 3$ because
  $q_1 \geq 3$. Hence $2a_2 \leq 2t + 3$. Since $t < a_2$, the only possibility is that $t = a_2 - 1$. By
  \autoref{lem:cases_2_3}, we may further assume that $m > (a_2 - 1)q_1 + 2$. Notice that this implies that
  $a_1 \geq 4$. The rest of the case $k = 1$ is handled by the next lemma.

  \begin{lemma}\label{lem:case_1}
    If $a_1 \geq 4$ and $(a_2 - 1)q_1 + 2 < m < a_2 q_1$, then $m$ is not the minimum abelian period of $w$.
  \end{lemma}
  \begin{proof}
    By \autoref{lem:larger_lower_bound}, we have $\norm{m\alpha} \geq \norm{\alpha} + (a_2 + 1)\norm{q_1\alpha}$. In
    order to conclude that $\abexp{m} < q_1 - 1$, it is enough, by the beginning of the proof of
    \autoref{lem:relation_exponent_convergent_2}, to verify the inequality \eqref{eq:foo3} for $k = 1$, that is, we
    need to show that
    \begin{equation*}
      q_0 + \alpha_2 + a_2 + 1 < (a_2 + 1)q_1.
    \end{equation*}
    Now $a_1 \geq 4$, so $(a_2 + 1)q_1 \geq 4a_2 + 4 > 3a_2 + \alpha_2 + 3 > q_0 + \alpha_2 + a_2 + 1$. Thus
    \autoref{lem:improved_ub} implies the claim.
  \end{proof}

  We have established all the cases, and we are now ready to prove \autoref{thm:main}.

  \begin{proof}[Proof of \autoref{thm:main}]
    Let $m$ be the minimum abelian period of a factor $w$ of slope $\alpha$. If $m < q_1$, then $m = t q_0$ with
    $1 \leq t < a_1$, that is, $m \in \qk{\alpha} \cup \mathcal{M}_\alpha$. Suppose that $m \geq q_1$. Then there
    exists a positive integer $k$ such that $q_k \leq m < q_{k+1}$. If $m \notin \qkl{\alpha} \cup \mathcal{M}_\alpha$,
    then Lemmas \ref{lem:almost_all}, \ref{lem:cases_2_3}, and \ref{lem:case_1} (together with the discussions
    preceding them), imply that $m$ cannot be the minimum abelian period of $w$. The conclusion is that
    $m \in \qkl{\alpha} \cup \mathcal{M}_\alpha$.
  \end{proof}

  Notice that \autoref{thm:main} directly implies \autoref{thm:fibonacci} because the slope of the Fibonacci word has
  continued fraction expansion $[0; 2, \overline{1}]$.

  Let us then see if \autoref{thm:main} completely characterizes the minimum abelian periods of factors of slope
  $\alpha$. We begin with the following proposition.

  \begin{proposition}\label{prp:qk_admissible}
    If $m \in \qk{\alpha}$, then there exists a factor of slope $\alpha$ having minimum abelian period $m$.
  \end{proposition}
  \begin{proof}
    Let $q_k \in \qk{\alpha}$ for some $k \geq 0$. Clearly the factor $0$ has minimum abelian period $q_0$, so we may
    assume that $k \geq 1$. We suppose that $0 \in I_0$ if $k$ is even and $0 \notin I_0$ otherwise. Consider the
    prefix $w$ of $\infw{s}_{0,\alpha}$ of length $\abexp{q_k}q_k$. The factor $w$ has abelian period $q_k$. Observe
    that $\abs{w} \geq q_{k+1} q_k$ because $\abexp{q_k} \geq q_{k+1}$ by \autoref{lem:exponent_lower_bound}. Let $m$
    be the minimum abelian period of $w$, and suppose for a contradiction that $m < q_k$. Then we have
    $\norm{m\alpha} \geq \norm{q_{k-1}\alpha}$ by the best approximation property. It follows by
    \autoref{lem:exponent_lower_bound} that $\abexp{m} < q_k + q_{k-1}$. Therefore
    \begin{equation*}
      q_{k+1} q_k \leq \abs{w} \leq (\abexp{m} + 2)m - 2 \leq (q_k + q_{k-1} + 1)(q_k - 1) - 2.
    \end{equation*}
    We thus obtain that
    \begin{equation*}
      ((a_{k+1} - 1)q_k - 1)q_k \leq -(q_k + q_{k-1} + 3).
    \end{equation*}
    The coefficient $(a_{k+1} - 1)q_k - 1$ of $q_k$ on the left is at least $-1$, which shows that the inequality is
    impossible. The conclusion is that $m = q_k$.
  \end{proof}

  \autoref{thm:main} and \autoref{prp:qk_admissible} together imply the following proposition. This was already
  implicitly present in \cite{2016:abelian_powers_and_repetitions_in_sturmian_words} as a corollary of
  \cite[Thm.~6.9]{2016:abelian_powers_and_repetitions_in_sturmian_words} and
  \cite[Thm.~6.12]{2016:abelian_powers_and_repetitions_in_sturmian_words}.

  \begin{proposition}
    The abelian period set of the Fibonacci word is the set of Fibonacci numbers.
  \end{proposition}

  We shall see next that it is not necessary that a denominator of a semiconvergent or a proper multiple of a denominator
  of a convergent is the minimum abelian period of some factor of slope $\alpha$. We begin with the following
  observation.

  \begin{proposition}\label{prp:fibonacci_counter_example}
    Let $k \geq 1$, and suppose that $a_{k+1} > 1$. Then there exists a factor of slope $\alpha$ whose minimum abelian
    period equals $q_{k+1,1}$ or $2q_k$.
  \end{proposition}
  \begin{proof}
    A suitable factor was essentially constructed in \autoref{sec:main_inequality}. Let us repeat the construction.
    Let $u$ be a factor of slope $\alpha$ that is a complete first return to $s$, the singular factor of length $q_k$,
    in the same phase. Such a factor exists by recurrence and \autoref{lem:return_times_singular}. Indeed, the return
    times of $s$, $q_{k+1}$ and $q_{k+2,1}$, equal $q_{k-1}$ modulo $q_k$, so $u$ has exactly $q_k + 1$ occurrences of
    $s$, all in different phases modulo $q_k$ except the first and final occurrence. Let $x$ be the final letter of
    $s$, and set $w = ux^{-1}$. It follows from what precedes that $w$ cannot have abelian period $q_k$. Moreover,
    $\abs{w} \geq (q_{k+1} + 1)q_k - 1$. Let $m$ be the minimal abelian period of $w$. Suppose first that $m < q_k$. By
    the best approximation property, we have $\norm{m\alpha} \geq \norm{q_{k-1}\alpha}$, so
    $\abexp{m} \leq \abexp{q_{k-1}}$. Now $\abexp{q_{k-1}} < q_k + q_{k-1}$ by \autoref{lem:exponent_lower_bound}, so
    $\abs{w} \leq (q_k + q_{k-1} + 1)(q_k - 1) - 2$. Therefore
    \begin{equation*}
      (q_{k+1} + 1)q_k - 1 \leq (q_k + q_{k-1} + 1)(q_k - 1) - 2,
    \end{equation*}
    which is equivalent to
    \begin{equation*}
      ((a_{k+1} - 1)q_k)q_k \leq -(q_k + q_{k-1} + 2).
    \end{equation*}
    This is a contradiction as the left side is clearly nonnegative. Thus we conclude that $m > q_k$. Write
    $s^{-1}us^{-1} = \beta_0 \dotsm \beta_{E-1}$ with $\abs{\beta_1} = \ldots = \abs{\beta_E} = q_k$. Observe that the
    words $\beta_0$, $\ldots$, $\beta_{E-1}$ are abelian equivalent and that $\beta_0 \dotsm \beta_{E-1}$ is a prefix
    of $s^{-1}w$. Let $\gamma_i = \beta_{2i}\beta_{2i+1}$ for $i = 0, \ldots, r$, where $r = \tfrac12(E - 2)$ if $E$ is
    even and $r = \tfrac12(E - 3)$ if $E$ is odd. The words $\gamma_0$, $\ldots$, $\gamma_r$ are abelian equivalent and
    have length $2q_k$. We may write $w = s\gamma_0 \dotsm \gamma_r v$, where $v = sx^{-1}$ if $E$ is even and
    $v = \beta_{E-1}sx^{-1}$ otherwise. The Parikh vector of $s$ is contained in the Parikh vector of $\gamma_0$ by
    \autoref{lem:singular_properties} (iii). Similarly \autoref{lem:singular_properties} (iii) shows that the Parikh
    vector of $sx^{-1}$ is contained in the Parikh vector of $\beta_0$. Therefore the Parikh vector of $v$ is contained
    in the Parikh vector of $\gamma_0$. Thus the word $w$ is an abelian repetition of period $2q_k$ with head $s$ and
    tail $v$. The conclusion is that $m \leq 2q_k$. Since $a_{k+1} > 1$, we have $q_{k+1} > 2q_k$. Hence
    \autoref{thm:main} implies that $m \in \{q_{k+2,1}, 2q_k\}$.
  \end{proof}

  Let us then see through examples that we cannot improve on \autoref{prp:fibonacci_counter_example}. Let
  $\alpha = [0; 2, 1, 2, 3, \overline{1}]$ ($\approx 0.3711$). Then the sequence of denominators of convergents is $2$,
  $3$, $8$, $27$, $\ldots$ and the sequence of denominators of semiconvergents is $5$, $11$, $19$, $\ldots$. Set
  $m = 2q_2 = 6$. It can be verified with the help of a computer that no factor of slope $\alpha$ has minimum abelian
  period $m$. Since $\abexp{m} = 6$, it is enough to compute the minimum abelian periods of factors up to length
  $(\abexp{m} + 2)m - 2 = 34$. In fact, the minimum abelian periods of factors up to length $34$ belong to the set
  $\{1, 2, 3, 5, 8\}$. Thus a proper multiple of a denominator of a convergent is not necessarily the minimum abelian
  period of some factor. Notice that the minimum abelian period $6$ is not ruled out by \eqref{eq:main_inequality}: the
  left side of \eqref{eq:main_inequality} equals $25$. There are thus other, unknown reasons why $6$ is not a minimum
  abelian period.

  It is possible to have a minimum abelian period of the form $tq_k$ with $t > 2$. Let
  $\alpha = [0; 2, 6, \overline{1}]$ ($\approx 0.4649$). Then the sequence of denominators of convergents is $2$, $13$,
  $15$, $\ldots$ and the sequence of denominators of semiconvergents is $3$, $5$, $7$, $9$, $11$, $\ldots$. The
  following factor of slope $\alpha$ of length $32$ has minimal abelian period $4q_1$:
  \begin{equation*}
    010100 \cdot 10101010 \cdot 10100101 \cdot 01010101 \cdot 01.
  \end{equation*}
  For the slope $[0; 2, 5, \overline{1}]$, no factor with minimum abelian period $8$ exists.

  Let finally $\alpha = [0; 2, 3, 2, \overline{1}]$ ($\approx 0.4355$). The sequence of denominators of convergents is
  $2$, $7$, $16$, $23$, $\ldots$ and the sequence of denominators of semiconvergents is $3$, $5$, $9$, $\ldots$. It can
  be verified that there is no factor of slope $\alpha$ with minimum abelian period $9$ (it is enough to study factors
  up to length $124$). Therefore a denominator of a semiconvergent is not necessarily the minimum abelian period of
  some factor. The possible abelian periods of factors up to length $124$ are in the set
  $\{1, 2, 3, 4, 5, 7, 14, 16\}$. The period $14$ is included as predicted by \autoref{prp:fibonacci_counter_example}.

  It seems to us that the problem of characterizing the possible abelian periods of factors of slope $\alpha$ is
  significantly harder than proving \autoref{thm:main}. The above examples indicate that the answer depends heavily on
  the arithmetic nature of the slope. We leave this problem open. Based on computer experiments, we have the following
  conjecture.

  \begin{conjecture}
    Let $\alpha = [0; \overline{2}]$. The abelian period set of a Sturmian word of slope $\alpha$ is
    $\qkl{\alpha} \cup \mathcal{M}_\alpha$.
  \end{conjecture}

  By \autoref{prp:fibonacci_counter_example}, there exists a factor with minimum abelian period that is not a
  denominator of a convergent whenever $a_k > 1$ for some $k \geq 2$. This gives the following interesting
  characterization of the Fibonacci subshift/the Golden ratio in terms of abelian periods.

  \begin{theorem}\label{thm:fibonacci_characterization}
    Let $\alpha$ be an irrational in $(0, \tfrac12)$. Then $\alpha = 1/\varphi^2$, where $\varphi$ is the Golden ratio,
    if and only if the minimum abelian period of every factor of slope $\alpha$ is in $\qk{\alpha}$.
  \end{theorem}
  \begin{proof}
    Say $\alpha = 1/\varphi^2$. Then $\alpha = [0; 2, \overline{1}]$ and, by \autoref{thm:main}, the abelian periods of
    factors of slope $\alpha$ are in $\qk{\alpha}$. Suppose then that the minimum abelian period of every factor of
    slope $\alpha$ is in $\qk{\alpha}$. \autoref{prp:fibonacci_counter_example} shows that there is necessarily a
    factor with minimum abelian period outside the set $\qk{\alpha}$ if $a_k > 1$ for some $k \geq 2$. Hence $a_k = 1$
    for all $k \geq 2$. If $a_1 > 2$, then the factor $01$ has abelian period $2$ and $2 \notin \qk{\alpha}$. Thus
    $a_1 = 2$. In other words, $\alpha = 1/\varphi^2$.
  \end{proof}

  \section{Remarks on \texorpdfstring{$k$}{k}-abelian Periods}\label{sec:k-abelian}
  In this section, we briefly discuss what changes if abelian equivalence is replaced by the more general $k$-abelian
  equivalence.

  Let $k$ be a positive integer. Two words $u$ and $v$ are \emph{$k$-abelian equivalent} if $\abs{u}_w = \abs{v}_w$ for
  each word $w$ of length at most $k$. Here $\abs{u}_w$ stands for the number of occurrences of $w$ as a factor of $u$.
  When $k = 1$, the $k$-abelian equivalence relation is simply the abelian equivalence relation. If $u$ and $v$ are
  $k$-abelian equivalent, then we write $u \sim_k v$. Notice that if $u \sim_k v$, then $u \sim_{k-1} v$ for all
  $k > 1$. For example, the words $0101100$ and $0011010$ are $2$-abelian equivalent, but they are not $3$-abelian
  equivalent. If $u = 0110$ and $v = 1101$, then $\abs{u}_w = \abs{v}_w$ for each word of length $2$, but
  $u \not\sim_2 v$ because $u$ and $v$ are not abelian equivalent. For words of length at least $k - 1$, we have
  $u \sim_k v$ if and only if $u$ and $v$ share a common prefix and a common suffix of length $k - 1$ and
  $\abs{u}_w = \abs{v}_w$ for each word $w$ of length $k$
  \cite[Lemma~2.4]{2013:on_a_generalization_of_abelian_equivalence_and_complexity_of_infinite}. For words of length at
  most $2k - 1$, the relation $\sim_k$ is in fact the equality relation $=$
  \cite[Lemma~2.4]{2013:on_a_generalization_of_abelian_equivalence_and_complexity_of_infinite}. If the words $u_0$,
  $\ldots$, $u_{e-1}$ are $k$-abelian equivalent, then their concatenation $u_0 \dotsm u_{e-1}$ is a \emph{$k$-abelian
  power} of period $\abs{u_0}$ and exponent $e$.

  The $k$-abelian equivalence is first introduced in the 1980 paper of J. Karhumäki
  \cite{1980:generalized_parikh_mappings_and_homomorphisms} in relation to the Post Correspondence Problem. The 2013
  paper \cite{2013:on_a_generalization_of_abelian_equivalence_and_complexity_of_infinite} by J. Karhumäki, A. Saarela,
  and L. Zamboni contains the first deeper study of $k$-abelian equivalence and, most importantly, the first research
  on $k$-abelian equivalence in relation to Sturmian words. One of their result is a characterization of Sturmian words
  as the aperiodic binary words whose factors of length $n$ belong to exactly $2k$ $k$-abelian equivalence classes if
  $n \geq 2k$ and to exactly $n + 1$ classes if $n \leq 2k - 1$
  \cite[Thm.~4.1]{2013:on_a_generalization_of_abelian_equivalence_and_complexity_of_infinite}. Another nice result is a
  general theorem from which it follows that Sturmian words contain $k$-abelian powers of arbitrarily high exponent
  \cite[Thm.~5.4]{2013:on_a_generalization_of_abelian_equivalence_and_complexity_of_infinite}. The results of Karhumäki
  et al. are made more precise in the paper \cite{2020:on_k-abelian_equivalence_and_generalized_lagrange_spectra} by
  the author and M. Whiteland where an approach based on continued fractions is developed to study $k$-abelian powers
  in Sturmian words. This approach yields results similar to those of
  \cite{2016:abelian_powers_and_repetitions_in_sturmian_words}. For example, the following analogue of
  \autoref{prp:exp_formula} concerning the maximum exponent $\abexp[k,\alpha]{m}$ of a $k$-abelian power of period $m$
  occurring in a Sturmian word of slope $\alpha$ is obtained. Here $\min L(2k-2)$ (resp. $\max L(2k-2)$) is the length
  of the shortest (resp. longest) interval among the intervals of factors of length $2k-2$.

  \begin{proposition}\cite[Lemma~3.10]{2020:on_k-abelian_equivalence_and_generalized_lagrange_spectra}
    Let $m$ be a positive integer and suppose that $\norm{m\alpha} < \min L(2k-2)$. Then
    \begin{equation*}
      \Abs{\Floor{\frac{\max L(2k-2)}{\norm{m\alpha}}} - \abexp[k,\alpha]{m}} \leq 1.
    \end{equation*}
  \end{proposition}

  Let us next discuss the generalization of an abelian period to this setting of $k$-abelian equivalence. The following
  definition is compatible with \autoref{def:abelian_period} when $k = 1$.

  \begin{definition}
    A word $w$ has \emph{$k$-abelian period} $m$ if $w$ is a factor of a $k$-abelian power of period $m$.
  \end{definition}

  \begin{example}\label{ex:k-ab_example}
    Let $w = 0100110$. The minimum abelian period of $w$ is $2$ because of the (only possible) factorization
    $w = 0 \cdot 10 \cdot 01 \cdot 10$. Since $10 \not\sim_2 01$, we see that the $2$-abelian period must be greater
    than $2$. The only candidate factorization of $w$ for $2$-abelian period $3$ is $w = 01 \cdot 001 \cdot 10$, but
    this is not good because no word of length $3$ beginning with $10$ can be $2$-abelian equivalent to $001$ (there
    must be a common prefix of length $1$). Keeping in mind the requirement for a common prefix and a common suffix of
    length $1$, we see that the relevant factorizations for period $4$ are $01 \cdot 0011 \cdot 0$ and
    $010 \cdot 0110$. The prefix $01$ of the first factorization cannot be completed to a word of length $4$ that is
    $2$-abelian equivalent to $0011$ since such a completion must begin with $0$ and then we are missing the factor
    $11$. By a similar analysis for the second factorization, we see that the minimum $2$-abelian period of $w$ is at
    least $5$. In fact, it is easily verified that it is $6$.
  \end{example}

  The abelian period can also be generalized in another way based directly on \autoref{def:abelian_period}.

  \begin{definition}
    Let $A$ be an alphabet and $k$ a positive integer. Suppose that $u_0$, $u_1$, $\ldots$, $u_t$ is an enumeration of
    the nonempty words over the alphabet $A$ of length at most $k$ in some fixed order. The \emph{generalized Parikh
    vector} $\Parikh[k]{w}$ of a word $w$ is the vector $(\abs{w}_{u_0}, \abs{w}_{u_1}, \ldots, \abs{w}_{u_t})$. If
    $u$ and $v$ are words, then we say that $\Parikh[k]{u}$ is \emph{contained} in $\Parikh[k]{v}$ if $\Parikh[k]{u}$
    is componentwise less than or equal to $\Parikh[k]{v}$.
  \end{definition}

  \begin{definition}
    A word $w$ over $A$ has \emph{$k$-abelian period $m$ in the second sense} if it is possible to write
    $w = u_0 u_1 \dotsm u_{n-1} u_n$ such that $n \geq 2$, $u_1 \sim_k \dotsm \sim_k u_{n-1}$, and $\Parikh[k]{u_0}$
    and $\Parikh[k]{u_n}$ are contained in $\Parikh[k]{u_1}$.
  \end{definition}

  This latter definition is different from the first one. For example, the minimum $2$-abelian period (in the second
  sense) of the word $w = 0100110$ of \autoref{ex:k-ab_example} is $4$ due to the factorization $w = 010 \cdot 0110$.
  It is difficult to argue which of the two generalizations is more natural. The author's opinion is that the former
  one is the right definition.

  \autoref{thm:main} does not have immediate consequences on minimal $k$-abelian periods of factors of Sturmian words.
  Indeed, as was seen in \autoref{ex:k-ab_example}, the minimum $k$-abelian period of a word might be larger than its
  minimum abelian period. The prefix $010010100$ of the Fibonacci word has minimum abelian period $2$ and minimum
  $2$-abelian period $5$ (in the first sense). No method presented in this paper is directly applicable to this more
  general setting, and hence we leave this problem open. It seems difficult to make any plausible conjecture in light
  of computer experiments. It would be natural to guess that direct analogues of \autoref{thm:main} and
  \autoref{thm:fibonacci_characterization} hold also in the $k$-abelian setting. Nonetheless, this is not true---at
  least not for all $k > 1$. It can be verified that the minimum $2$-abelian and $3$-abelian periods of each factor of
  the Fibonacci word of length at most $200$ are Fibonacci numbers. The same seems to hold for $k$-abelian periods when
  $k = 4$, $\ldots$, $6$, but for $k = 7$, the situation is different. The factor $01001001010010010100101$ of the
  Fibonacci word has minimum $7$-abelian period $16$. Curiously, if the definition in the second sense is used, then
  the minimum $2$-abelian periods of the factors of the Fibonacci word of length at most $200$ are exactly $1$, $2$,
  $3$, $4$, $5$, $8$, $13$, $21$. The factor $0010100$ indeed has minimum $2$-abelian period $4$ in the second sense.
  The corresponding minimum $3$-abelian periods are $1$, $2$, $3$, $5$, $6$, $7$, $8$, $10$, $13$, $21$.

  \section*{Acknowledgments}
  The author thanks M. Whiteland for valuable discussions.

  \printbibliography
  
\end{document}